\documentclass[12pt,draftcls,onecolumn]{IEEEtran}
\usepackage{amsfonts}
\usepackage[tbtags]{amsmath}
\usepackage{graphicx,subfigure}
\usepackage{color}      
\usepackage{epsfig}
\usepackage{amssymb}
\usepackage{tipa}
\usepackage{bbm}

\long\def\comment#1{}

\newcommand{\beq}{\begin{equation}}
\newcommand{\eeq}{\end{equation}}
\newcommand{\beqno}{\begin{equation*}}
\newcommand{\eeqno}{\end{equation*}}

\newcommand{\bes}{\begin{split}}
\newcommand{\ees}{\end{split}}
\newcommand{\bdm}{\begin{displaymath}}
\newcommand{\edm}{\end{displaymath}}

\newtheorem{theorem}{Theorem}
\newtheorem{claim}{Claim}
\newtheorem{lemma}{Lemma}

\newtheorem{definition}{Definition}

\newcommand{\bd}{\begin{definition}}
\newcommand{\ed}{\end{definition}}

\newcommand{\bv}{\begin{vugraph}}
\newcommand{\ev}{\end{vugraph}}
\newcommand{\bi}{\begin{itemize}}
\newcommand{\ei}{\end{itemize}}
\newcommand{\ben}{\begin{enumerate}}
\newcommand{\een}{\end{enumerate}}

\newcommand{\bean}{\begin{eqnarray*} }
\newcommand{\eean}{\end{eqnarray*} }
\newcommand{\bea}{\begin{eqnarray} }
\newcommand{\eea}{\end{eqnarray} }

\newcommand{\ba}{\begin{array} }
\newcommand{\ea}{\end{array} }

\begin{document}

\title{Fast Algorithm for Finding Unicast Capacity of Linear Deterministic Wireless Relay Networks}

\author{\authorblockN{Cuizhu Shi and Aditya Ramamoorthy}\\
\vspace{-3mm}
\authorblockA{Department of Electrical and Computer Engineering, Iowa State University, Ames, Iowa 50011\\
Email: \{cshi, adityar\}@iastate.edu}
}

\maketitle

\begin{abstract}

The deterministic channel model for wireless relay networks proposed by Avestimehr, Diggavi and Tse `07 has captured the broadcast and inference nature of wireless communications and has been widely used in approximating the capacity of wireless relay networks. The authors generalized the max-flow min-cut theorem to the linear deterministic wireless relay networks and characterized the unicast capacity of such deterministic network as the minimum rank of all the binary adjacency matrices describing source-destination cuts whose number grows exponentially with the size of the network.
In this paper, we developed a fast algorithm for finding the unicast capacity of a linear deterministic wireless relay network by finding the maximum number of linearly independent paths using the idea of path augmentation. We developed a modified depth-first search algorithm
tailored for the linear deterministic relay networks for finding linearly independent paths whose total number proved to equal the unicast capacity of the underlying network. The result of our algorithm suggests a capacity-achieving transmission strategy with one-bit length linear encoding at the relay nodes in the concerned linear deterministic wireless relay network. The correctness of our algorithm for universal cases is given by our proof in the paper. Moreover, our algorithm has a computational complexity bounded by $O(|{\cal{V}}_x|\cdot C^4+d\cdot |{\cal{V}}_x|\cdot C^3)$ which shows a significant improvement over the previous results for solving the same problem by Amaudruz and Fragouli (whose complexity is bounded by $O(M\cdot |{\cal{E}}|\cdot C^5)$ with $M\geq d$ and $|{\cal{E}}|\geq|{\cal{V}}_x|$) and by Yazdi and Savari (whose complexity is bounded by $O(L^8\cdot M^{12}\cdot h_0^3+L\cdot M^6\cdot C\cdot h_0^4)$ with $h_0\geq C$).

\end{abstract}

\section{Introduction\label{introduction}}

The complex signal interactions in wireless relay networks challenge the study of wireless information flow for many years.
To characterize the capacity and capacity-achieving transmission schemes for wireless relay networks still remain open questions. Towards this end, the deterministic channel model for wireless relay networks proposed by Avestimehr, Diggavi and Tse \cite{amir2007_deterministicmodel} has been a significant progress.
The broadcast and inference are two fundamental features of wireless communications. The deterministic channel model captures the broadcast and inference features of wireless communications in addition to converting the wireless relay networks into deterministic networks. Studying the information flow in the deterministic networks provides a way to find out the approximated capacity and corresponding transmission strategies for original wireless relay networks.

Gaussian channel has been the most widely used channel model for the link channels in wireless relay networks. The deterministic channel model quantizes the transmitted signal into different bit levels and at the receiver keeps the signal bit levels above the noise level (which depends on the signal to noise ratio (SNR) of the channel) so as to convert the original Gaussian channel into a deterministic channel without random noise variables. The broadcasting of signal at the transmitter is still preserved in the deterministic channel and the interference of signal at the receiver is modeled by modulo two sum of the bits arrived at the same signal level.

Now we introduce the deterministic channel model by using the example of a point-to-point AWGN channel from \cite{amir2007_deterministicmodel}. Consider an AWGN channel $y=hx+z$ where $z\sim{\cal{N}}(0,1)$ (${\cal{N}}$ represents Gaussian distribution) and $h=\sqrt{{\mbox{SNR}}}$. Assume $x$ and $z$ are real numbers, then we can write $y\approx 2^n\sum_{i=1}^nx(i)2^{-i}+\sum_{i=1}^\infty (x(i+n)+z(i))2^{-i}$ where $n=\lceil\frac{1}{2}\log {\mbox{SNR}}\rceil$. If we think of the transmitted signal $x$ as a sequence of bits at different signal levels, then the deterministic channel model truncates $x$ and passes only its bits above noise level. Fig. \ref{fig1} gives a concrete example. At the transmitter node $T_x$ and receiver node $R_x$, each small cycle represents a signal level. Assume $n=4$, so only the first four most significant signal levels or bits from $T_x$ are received at $R_x$. Accordingly each edge in the model can transmit one-bit information at a time.

\begin{figure}[htbp]
\centering
\includegraphics[scale=0.4]{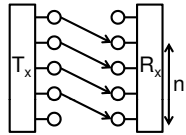}
\caption[deterministic model]{An example of the deterministic channel model for a point-to-point Gaussian channel. }
\label{fig1}
\end{figure}

The deterministic channel model we discussed above is called linear finite-field deterministic channel model in \cite{amir2007_wirelessnetworkinfoflow}, which is referred to as linear deterministic channel model in this paper. In \cite{amir2007_deterministicmodel}\cite{amir2007_wirelessnetworkinfoflow}, the unicast (i.e., with one source $S$ and one destination $D$) capacity $C$ of the linear deterministic wireless relay networks was characterized as the minimum rank of all the binary adjacency matrices describing $S-D$ cuts. Since the total number of such cuts grows exponentially with the size of the network, an exhaustive search for the minimum rank of the adjacency matrix for these cuts results in an algorithm with complexity exponential in the size of the network. A polynomial-time algorithm is desirable for finding the unicast capacity of the linear deterministic wireless relay networks. We will see later that the path augmentation algorithms for graphs cannot be directly applied here because the definitions for the cut value for two cases are different (see the definition for the cut value in a linear deterministic wireless relay network in Section \ref{notation}). In this paper, we only consider linear deterministic wireless relay networks. Since an arbitrary deterministic wireless relay network can be unfolded over time to create a layered deterministic network through time-expansion technique \cite{amir2007_wirelessnetworkinfoflow}, our algorithm developed in this paper for layered networks also applies to general networks.

In \cite{aurore2009_combinatorial_algo_deterministic}, Amaudruz and Fragouli proposed a polynomial-time algorithm for finding the unicast capacity of a linear deterministic wireless relay network by trying to identify the maximum number of linearly independent paths in the network using the idea of path augmentation. In \cite{sadegh2009_combinatorialstudyofdeterministic}, Yazdi and Savari developed a two-dimensional Rado-Hall transversal theorem for block matrices and used the submodularity of the capacity of a cut to formulate the problem as a linear program over the intersection of two polymatroids so as to solve the problem in polynomial time. Compared with these previous results, our algorithm has a significantly less computational complexity as explained in Section \ref{complexity}. Moreover, our algorithm comes with a more intuitive understanding as we explained in Section \ref{ouralgo}.

The paper is organized as follows. In Section \ref{notation}, we introduce the notations and definitions used throughout this paper. Section \ref{ouralgo} gives a detailed description of our algorithm for finding the unicast capacity of any given linear deterministic wireless relay network. Section \ref{analysis} is about the algorithm analysis including the proof of correctness and complexity analysis.
Section \ref{conclusion} concludes the paper.

\section{Notations and Definitions\label{notation}}

Let ${\cal{G=(N,V,E)}}$ be a layered deterministic wireless relay network \cite{amir2007_wirelessnetworkinfoflow} where ${\cal{N}}$ is the set of super nodes referring to the nodes in the original wireless relay network, ${\cal{V}}$ is the set of nodes referring to different signal levels incident to super nodes and ${\cal{E}}$ is the set of directed edges going from node to node. Denote ${\cal{V=V}}_x\cup {\cal{V}}_y$ with ${\cal{V}}_x=\{x: (x,y)\in {\cal{E}}\}$ and ${\cal{V}}_y=\{y: (x,y)\in {\cal{E}}\}$. We shall call ${\cal{V}}_x$ as transmitting nodes and ${\cal{V}}_y$ as receiving nodes. For example, in Fig. \ref{fig1}, ${\cal{N}}=\{T_x, R_x\}$ and $T_x$ has five transmitting nodes and $R_x$ has five receiving nodes. In a layered network, all paths from the source to the destination have equal lengths \cite{amir2007_wirelessnetworkinfoflow}, so we can divide the set of super nodes into different layers each layer containing only super nodes with the same distance to the source. Assume there are $L$ layers of super nodes in ${\cal{G}}$ and $M$ is the maximum number of super nodes in each layer. The source super node $S$ consists the first layer and the destination super node $D$ consists the last layer. Let ${\cal{N}}(x_i)$ (or ${\cal{N}}(y_j)$) denote the super node where a transmitting node $x_i$ (or a receiving node $y_j$) belongs to. Let ${\cal{L}}(N)$ (or ${\cal{L}}(x_i)$, ${\cal{L}}(y_j)$) denote the layer number where super node $N$ (or $x_i$, $y_j$) belongs to. In a layered network, if $(x,y)\in {\cal{E}}$, $(x',y')\in {\cal{E}}$ and ${\cal{L}}(x)={\cal{L}}(x')$, then we must have ${\cal{L}}(y)={\cal{L}}(y')={\cal{L}}(x)+1={\cal{L}}(x')+1$.

A cut, $\Omega$, in a deterministic relay network ${\cal{G}}$ is a partition of the super nodes ${\cal{N}}$ (together with their incident nodes) into two disjoint sets $\Omega$ and $\Omega^c$ such that $S\in\Omega$ and $D\in\Omega^c$.
For convenience, we call a cut a layer cut if all edges across the cut are emanating from nodes belonging to the same layer, otherwise we call it a cross-layer cut. We say an edge $(x_i,y_j)\in {\cal{E}}$ belongs to layer cut $l$ if ${\cal{L}}(x_i)=l$. For a layered deterministic network, there are exactly $L-1$ layer cuts.

The adjacency matrix $T$ for a set of transmitting nodes $\textbf{x}=\{x_1,x_2,...x_m\}, x_i\in {\cal{V}}_x$, and a set of receiving nodes $\textbf{y}=\{y_1,y_2,...y_n\},y_i\in {\cal{V}}_y$ in a deterministic relay network is a binary matrix of size $m\times n$ with rows corresponding to $\{x_i,x_i\in \textbf{x}\}$ and columns corresponding to $\{y_i,y_i\in\textbf{y}\}$ and $T(i,j)=1$ if $(x_i,y_j)\in {\cal{E}}$.
The adjacency matrix for a set of $k$ edges is the binary adjacency matrix for the set of their transmitting nodes and the set of their receiving nodes.

A set of $k$ edges are said to be linearly independent if the adjacency matrix for them has rank $k$, otherwise they are said to be linearly dependent. In a layered deterministic network, each $S-D$ path is of length $L-1$ and crosses each layer cut exactly once. A set of $k$ $S-D$ paths are said to be linearly independent if each set of their edges of size $k$ crossing each layer cut are linearly independent, otherwise they are said to be linearly dependent. Lemma \ref{lemm9} shows that $k$ linearly independent $S-D$ paths correspond to a transmission scheme of rate $k$.

In a deterministic wireless relay network, there are intermediate super nodes (exclude $S$ and $D$) corresponding to the relay nodes in the original wireless relay network which have both transmitting nodes and receiving nodes. It is shown in Lemma \ref{lemm9} and Theorem \ref{theorem1} that one-bit length linear encoding functions at the relay super nodes are sufficient for constructing capacity-achieving transmission schemes between $S$ and $D$ for the underlying linear deterministic relay network.

Let ${\cal{E}}_{\Omega}$ be the set of edges crossing the cut $\Omega$ in a linear deterministic relay network. The cut value of $\Omega$ in the linear deterministic relay network is defined as the rank of the binary adjacency matrix for ${\cal{E}}_{\Omega}$, which equals the number of linearly independent edges in ${\cal{E}}_{\Omega}$. Note that here the cut value defined for linear deterministic wireless relay networks is different from that for graphs (which is just the number of edges crossing the cut). It is proved \cite{amir2007_deterministicmodel}\cite{amir2007_wirelessnetworkinfoflow} that the unicast capacity of a linear deterministic wireless relay network is equal to the minimum cut value among all cuts.

\section{Our Algorithm\label{ouralgo}}

\subsection{Algorithm Outline}\label{description}

The max-flow min-cut theorem has been generalized to the linear deterministic wireless relay networks, but since the definitions of cut value for graphs (the number of edges crossing the cut) and for linear deterministic relay networks (the number of linearly independent edges crossing the cut) are different, the path augmentation algorithms for the latter are different from that for the former. As we will see later from this paper, the similarity is that in both cases the path augmentation algorithms operate in iterations and in each iteration they complete an additional path. The difference is that contrast to the simple addition of an available edge to a path in the path augmentation algorithms for graphs, the addition of edges to a path in a deterministic network has to satisfy some rank requirement of the adjacency matrix to avoid linear dependence among used path edges in each layer cut.

Our algorithm is basically a path augmentation algorithm for finding the maximum number $K$ of linearly independent $S-D$ paths in a layered linear deterministic relay network, where $K$ proves to be equal to $C$ in Section \ref{proof}. Moreover, the $K$ identified paths correspond to a capacity-achieving transmission strategy for the underlying deterministic network.

The algorithm operates in iterations. In iteration $k$, a modified depth-first search (MDFS) algorithm tailored for the linear deterministic relay networks is carried out on the graph ${\cal{G}}$ trying to find an $S-D$ path in addition to the $k-1$ paths found in the first $k-1$ iterations (stored in a structure ${\cal{P}}$) so that the $k$ found paths (stored in a structure ${\cal{P}}'$) are linearly independent. If MDFS returns True indicating that the $k$th $S-D$ path is found, then totally $k$ linearly independent paths have been found and the algorithm continues to iteration $k+1$. If it returns False indicating that no $S-D$ path is found, then no more independent path exists and the algorithm stops while those identified paths in ${\cal{P}}'$ suggest a capacity-achieving transmission scheme.

\subsection{Preliminaries}\label{definitions}

The following notations or structures are used in our algorithm. ${\cal{P}}$ and ${\cal{P}}'$ are structures storing information about the $k-1$ paths found in previous $k-1$ iterations and information about the updated $k-1$ paths and the partial path found in the current $k$th iteration respectively. Denote ${\cal{E}}_u$ (${\cal{E}}_u^i$) as the set of edges in ${\cal{E}}$ used by paths in ${\cal{P}}'$ (in layer cut $i$). The MDFS algorithm ensures that edges in ${\cal{E}}_u^i$ are linearly independent, i.e., rank$(T({\cal{E}}_u^i))=|{\cal{E}}_u^i|$.
Denote ${\cal{V}}_{xu}^{i}=\{x: (x,y)\in{\cal{E}}_u^i\}$ and ${\cal{V}}_{yu}^{i}=\{y: (x,y)\in{\cal{E}}_u^i\}$.
${\cal{E}}_u$ (${\cal{E}}_u^i$), ${\cal{V}}_{xu}^{i}$, ${\cal{V}}_{yu}^{i}$ and ${\cal{P}}'$ are initialized at the beginning of every iteration according to ${\cal{P}}$ and maintained by MDFS in the current iteration subject to changes. As we will notice in Section \ref{mdfs}, each move MDFS makes ensures that rank$(T({\cal{E}}_u^i))=|{\cal{E}}_u^i|$. Let ${\cal{E}}_{u,x_j}^{i}$, ${\cal{V}}_{xu,x_j}^{i}$ and ${\cal{V}}_{yu,x_j}^{i}$ denote the instantaneous sets of ${\cal{E}}_u^i$, ${\cal{V}}_{xu}^i$ and ${\cal{V}}_{yu}^i$ when a transmitting node $x_j$ is being explored in layer cut $i={\cal{L}}(x_j)$ by MDFS in a certain iteration. Furthermore, denote ${\cal{V}}_{x,N}=\{x: x\in {\cal{V}}_x$ and ${\cal{N}}(x)=N\}$ and ${\cal{V}}_{y,N}=\{y: y\in {\cal{V}}_y$ and ${\cal{N}}(y)=N\}$. Denote ${\cal{E}}_{{\cal{P}}}=\{e: e\in{\cal{E}}, e$ used by some paths in ${\cal{P}}\}$ and ${\cal{V}}_{x{\cal{P}}}=\{x: (x,y)\in {\cal{E}}_{{\cal{P}}}$ for some $y\}$ and ${\cal{V}}_{y{\cal{P}}}=\{y: (x,y)\in {\cal{E}}_{{\cal{P}}}$ for some $x\}$.


\begin{definition}
\textbf{${\cal{V}}_{xspan}^{x_j}$}: We define ${\cal{V}}_{xspan}^{x_j}$ as the set satisfies ${\cal{V}}_{xspan}^{x_j}\subseteq{\cal{V}}_{xu,x_j}^{i}$ and
\begin{eqnarray}\label{eqn21}
T(x_j,{\cal{V}}_{yu,x_j}^{i})=\sum_{x\in{\cal{V}}_{xspan}^{x_j}}T(x,{\cal{V}}_{yu,x_j}^{i})
\end{eqnarray}
Since $T({\cal{V}}_{xu,x_j}^{i},{\cal{V}}_{yu,x_j}^{i})$ has full rank, ${\cal{V}}_{xspan}^{x_j}$ is the unique such set.
\end{definition}

Let function Span$(x_j)$ be the function for computing ${\cal{V}}_{xspan}^{x_j}$.
In iteration $k+1$ of our algorithm, $|{\cal{V}}_{xu,x_j}^{i}|=k$.

\begin{lemma}\label{lemma1}
Let $|{\cal{V}}_{xu,x_j}^{i}|=k=|{\cal{V}}_{yu,x_j}^{i}|$. The computational complexity of Span$(x_j)$ for finding the set ${\cal{V}}_{xspan}^{x_j}$ is bounded by $O(k^3)$. For $\forall x\in{\cal{V}}_{xspan}^{x_j}$, rank$(T({\cal{V}}_{xu,x_j}^{i},{\cal{V}}_{yu,x_j}^{i}))=$ rank$(T({\cal{V}}_{xu,x_j}^{i}+x_j-x,{\cal{V}}_{yu,x_j}^{i}))=k$ and rank$(T({\cal{V}}_{xspan}^{x_j}+x_j,{\cal{V}}_{yu,x_j}^{i}))=$ rank$(T({\cal{V}}_{xspan}^{x_j},{\cal{V}}_{yu,x_j}^{i}))=$ rank$(T({\cal{V}}_{xspan}^{x_j}+x_j-x,{\cal{V}}_{yu,x_j}^{i}))=|{\cal{V}}_{xspan}^{x_j}|\leq k$.
\end{lemma}
\begin{proof}
To solve the set ${\cal{V}}_{xspan}^{x_j}$ is equivalent to solving the system of linear equations, ${\cal{V}}_{xspan}^{x_j}\cdot T({\cal{V}}_{xu,x_j}^{i},{\cal{V}}_{yu,x_j}^{i})=T(x_j,{\cal{V}}_{yu,x_j}^{i})$ in $GF(2)$ which can be accomplished in time $O(k^3)$ by using Gaussian elimination. The second statement is obvious.
\end{proof}

Consider the subgraph consisting of nodes $x_j\cup{\cal{V}}_{xu,x_j}^i\cup{\cal{V}}_{yu,x_j}^i$ and the edges connecting them in ${\cal{G}}$. Let ${\cal{G}}_{sub}^{x_j}$ denote the graph obtained by reversing the directions of the edges in ${\cal{E}}_{xu,x_j}^i$ in the above subgraph.

\begin{lemma}\label{lemma5}
There is a directed path from $x_j$ to any $x\in{\cal{V}}_{xspan}^{x_j}$ in graph ${\cal{G}}_{sub}^{x_j}$. Let FindIndPaths be the function for finding out all these $|{\cal{V}}_{xspan}^{x_j}|$ paths from $x_j$. The computational complexity of FindIndPaths$(x_j,{\cal{V}}_{xspan}^{x_j})$ is bounded by $O(k^2)$ in iteration $k$.
\end{lemma}
\begin{proof}
From Lemma \ref{lemma1}, for $\forall x\in{\cal{V}}_{xspan}^{x_j}$, rank$(T({\cal{V}}_{xu,x_j}^{i},{\cal{V}}_{yu,x_j}^{i}))=$ rank$(T({\cal{V}}_{xu,x_j}^{i}+x_j-x,{\cal{V}}_{yu,x_j}^{i}))=k$ where $k=|{\cal{P}}|$. Introduce an auxiliary receiving node $y'$ and an edge $(x,y')$. It's easy to see that rank$(T({\cal{V}}_{xu,x_j}^{i}+x_j,{\cal{V}}_{yu,x_j}^{i}+y'))=k+1$. Given rank$(T({\cal{V}}_{xu,x_j}^{i},{\cal{V}}_{yu,x_j}^{i}))=k$ and rank$(T({\cal{V}}_{xu,x_j}^{i}+x_j,{\cal{V}}_{yu,x_j}^{i}+y'))=k+1$, there is a size $k$ prefect matching between ${\cal{V}}_{xu,x_j}^{i}$ and ${\cal{V}}_{yu,x_j}^{i}$, $M_1={\cal{E}}_{u,x_j}^i$ being such a matching, and a size $k+1$ perfect matching between ${\cal{V}}_{xu,x_j}^{i}+x_j$ and ${\cal{V}}_{yu,x_j}^{i}+y'$. Using a similar argument as in finding the maximum bipartite matching, we claim that there is an alternating path, relative to the matching $M_1$, starting from an unused transmitting node $x_j$ to an unused receiving node $y'$, alternating between edges not in the current matching $M_1$ and edges in the current matching $M_1$, i.e., there is a path $P_{x_j\rightarrow y'}=\{(x_j,y_1),(y_1,x_1),(x_1,y_2),(y_2,x_2),...(x_{m-1},y_m),(y_m,x_m),(x_m,y')=(x,y')\}$ with $(x_i,y_i),1\leq i\leq m$ being edges in $M_1={\cal{E}}_{u,x_j}^i$. So we proved that there is a path $P_{x_j\rightarrow x}=\{(x_j,y_1),(y_1,x_1),(x_1,y_2),(y_2,x_2),...(x_{m-1},y_m),(y_m,x_m)=(y_m,x)\}$ with $(x_i,y_i),1\leq i\leq m$ being edges in ${\cal{E}}_{u,x_j}^i$. Without loss of generality, we also use the following representation of the path $P_{x_j\rightarrow x}=\{(x_j,y_1),(x_1,y_1),(x_1,y_2),(x_2,y_2),...(x_{m-1},y_m),(x_m,y_m)=(x,y_m)\}$ with $(x_i,y_i),1\leq i\leq m$ being edges in ${\cal{E}}_{u,x_j}^i$.

In iteration $k$ of our algorithm, $|{\cal{V}}_{xu}^i|=|{\cal{V}}_{yu}^i|=k-1$, so the number of nodes in ${\cal{G}}_{sub}^{x_j}$ is bounded by $O(k)$, which also means the number of edges in ${\cal{G}}_{sub}^{x_j}$ is bounded by $O(k^2)$. To find directed paths from $x_j$ to all $x\in{\cal{V}}_{xspan}^{x_j}$ in ${\cal{G}}_{sub}^{x_j}$ takes time bounded by $O(k^2)$ by using some well-known graph traversal algorithms, like breadth-first search.
\end{proof}

Let's briefly recall the depth-first search (DFS) algorithm first. DFS algorithm is a well-known algorithm for traversing graphs. To find a path, DFS progresses by exploring the outgoing edges of a node before exploring any other outgoing edges of the node's predecessor. In this way, it proceeds deeper and deeper until it reaches the goal node when it stops or until it reaches a node without any outgoing edges when it backtracks to the most recent node that it hasn't finished exploring. In the search process, each node in the graph would be explored at most once.

Our modified depth-first search (MDFS) algorithm for linear deterministic relay networks inherits the basic forwarding and backtracking operations from the basic DFS algorithm and each super node in ${\cal{N}}$ is treated like a node in DFS. The difference is that in order to find an $S-D$ path in the graph ${\cal{G}}=({\cal{N,V,E}})$ linearly independent to the paths in ${\cal{P}}$, we have to avoid linear dependency between different paths in our algorithm, which means that instead of allowing a valid forwarding move along each outgoing edge when a node is explored as in DFS, more constraints should be imposed on the forwarding moves of a super node in MDFS. We propose the following MDFS algorithm to accomplish the path augmentation task in our problem.

\subsection{Modified Depth-First Search (MDFS) Algorithm}\label{mdfs}

In the following, the exploration to a super node $N\in{\cal{N}}$ or a node $v\in{\cal{V}}$ by MDFS refers to that MDFS has extended the path found in the current iteration to $N$ or ${\cal{N}}(v)$ and now it continues to extend the path further from $N$ or from $v$. The exploration to a super node $N$ is realized as the exploration to some of its unexplored incident nodes $v$ (which we call admissible nodes of $N$ in Definition \ref{definition1}). Once a super node or a node is explored, it's labeled explored. During the running time of MDFS, each super node $N\in{\cal{N}}$ could be labeled as unexplored or explored indicating that it hasn't or has been explored by MDFS. Each node in ${\cal{V}}$ could be labeled as unexplored or explored indicating that it doesn't allow or allows exploration by MDFS. There is also a type labeling with each node in ${\cal{V}}_x$ indicating how it can be explored by MDFS if it allows exploration. Let SetLabel$(X,$ LABEL) be the function setting label of $X$ ($X$ can be a super node or a node) as LABEL (LABEL can be explored or unexplored). Let LABEL$=$GetLabel$(X)$ be the function returning the label of $X$. Let SetType$(x,B)$ be the function setting type of transmitting node $x$ as $B$ ($B$ can be $1,2,$ or $3$). Let $B=$GetType$(x)$ be the function returning the type of $x$.

\begin{definition}\label{definition1}
\textbf{Admissible nodes and admissible forwarding moves of super nodes --} We define the admissible nodes for a super node $N$ when $N$ is explored by MDFS as follows: ${\cal{V}}_{N}^{ad}=\{x: x\in{\cal{V}}_{x,N}, x\not\in {\cal{V}}_{xu}^{{\cal{L}}(N)}, x$ labeled unexplored or $y: y\in{\cal{V}}_{y,N}, y\in{\cal{V}}_{y{\cal{P}}}, y$ labeled unexplored$\}$. The exploration to a super node $N$ is realized as the exploration to its admissible nodes. The following forwarding moves starting from ${\cal{V}}_{N}^{ad}$ are defined as admissible forwarding moves allowed in MDFS algorithm on the graph ${\cal{G}}=({\cal{N,V,E}})$ when it explores a super node $N$.
\begin{enumerate}
\item
\emph{Type $1$}: moving forward along $(x,y)\in{\cal{E}}$ with $x\in{\cal{V}}_{N}^{ad}$, GetType$(x)=1,2,$ or $3$ and $y\not\in{\cal{V}}_{yu}^{{\cal{L}}(N)}$, rank$(T({\cal{V}}_{xu,x}^{{\cal{L}}(N)}+x,{\cal{V}}_{yu,x}^{{\cal{L}}(N)}+y))=|{\cal{P}}|+1$ and GetLabel$({\cal{N}}(y))=$unexplored.
\item
\emph{Type $2$}: moving forward along the path $P_{x\rightarrow x'}$ (as proved to exist in Lemma \ref{lemma5}) for any $x'\in{\cal{V}}_{xspan}^{x}$ with $x\in{\cal{V}}_{N}^{ad}$, GetType$(x)=1$ or $3$.
\item
\emph{Type $3$}: moving forward along $(x,y)\in{\cal{E}}$ with $y\in{\cal{V}}_{N}^{ad}$, $(x,y)\in{\cal{E}}_{u,x}^{{\cal{L}}(N)-1}$.
\end{enumerate}
When $N$ is explored, $x\in{\cal{V}}_{N}^{ad}$ with GetType$(x)=1$ or $3$ would allow to start type $1$ and type $2$ admissible forwarding moves, $x\in{\cal{V}}_{N}^{ad}$ with GetType$(x)=2$ would allow to start type $1$ admissible forwarding moves only and $y\in{\cal{V}}_{N}^{ad}$ would allow to start type $3$ admissible forward moves. There is an ordering in exploring ${\cal{V}}_{N}^{ad}$: type $2$ and type $3$ nodes should be explored first, then type $1$ nodes and finally the receiving nodes.
\end{definition}

\begin{definition}\label{definition2}
\textbf{Modified depth-first search algorithm (MDFS):} The MDFS algorithm is defined in terms of initialization, exploring of a super node, labeling of ${\cal{N,V}}$ and updating of ${\cal{P}}'$ and ${\cal{E}}_u$
as follows:
\begin{enumerate}
\item
\emph{Initialization: }Set ${\cal{N}}$ and ${\cal{V}}$ as unexplored and ${\cal{V}}_x$ as type $1$ nodes.
\item
\emph{Exploring of a super node $N$: }When MDFS explores a super node $N$, it means that a partial path $P_{|{\cal{P}}|+1}$ from $S$ to $N$ has been found in addition to the $|{\cal{P}}|$ complete $S-D$ paths. As mentioned in Definition \ref{definition1}, the exploration to a super node $N$ is realized as exploration to its admissible nodes ${\cal{V}}_{N}^{ad}$ and three kinds of admissible forwarding moves from ${\cal{V}}_{N}^{ad}$ are allowed. There is an ordering in exploring ${\cal{V}}_{N}^{ad}$: type $2$ and type $3$ nodes should be explored first, then type $1$ nodes and finally the receiving nodes. Once a super node or a node is explored by MDFS, it is labeled explored. Now we explain how to understand these three admissible forwarding moves and how to label ${\cal{N,V}}$ and update ${\cal{P}}'$, ${\cal{E}}_u$ with these moves.
\begin{enumerate}
\item
Type $1$ admissible forwarding move: A type $1$ admissible forwarding move can be understood as that MDFS extends $P_{|{\cal{P}}|+1}$ along an edge $(x,y)$ from super node $N$ to super node ${\cal{N}}(y)$, i.e., $P_{|{\cal{P}}|+1}=P_{|{\cal{P}}|+1}+(x,y)$. ${\cal{P}}'$ is updated accordingly. Then MDFS makes ${\cal{N}}(y)$ with ${\cal{L(N}}(y))={\cal{L}}(N)+1$ the next super node to be explored when ${\cal{E}}_u^{{\cal{L}}(N)}$ is updated as ${\cal{E}}_u^{{\cal{L}}(N)}+(x,y)$ and $|{\cal{E}}_u^{{\cal{L}}(N)}|$ increases from $|{\cal{P}}|$ to $|{\cal{P}}|+1$.
\item
Type $2$ admissible forwarding move: A type $2$ admissible forwarding move can be understood as that MDFS updates ${\cal{P}}'$ according to the path $P_{x\rightarrow x'}$ (proved to exist in Lemma \ref{lemma5}) for some $x'\in{\cal{V}}_{xspan}^{x}$ as follows. Let $P_{x\rightarrow x'}=\{(x,y_1)$, $(x_1,y_1)$, $(x_1,y_2)$, $...(x_m,y_m)=(x',y_m)\}$ with $(x_i,y_i),1\leq i\leq m$ being path edges used by $m$ $S-D$ paths in ${\cal{P}}'$
existing before the current move, denoted as $P_i,1\leq i\leq m$. After the current type $2$ forwarding move, $P_i,1\leq i\leq m$ and $P_{|{\cal{P}}|+1}$ are updated to $P_i'$ and $P_{|{\cal{P}}|+1}'$ as follows. Let $P_i(N_1,N_2)$ denote the segment of $P_i$ from super node $N_1$ to $N_2$. $P_1'=P_{|{\cal{P}}|+1}(S,{\cal{N}}(x))+(x,y_1)+P_1({\cal{N}}(y_1),D)$, $P_i'=P_{i-1}(S,{\cal{N}}(x_{i-1}))+(x_{i-1},y_i)+P_i({\cal{N}}(y_i),D),1<i\leq m$ and $P_{|{\cal{P}}|+1}'=P_m(S,{\cal{N}}(x'))$. ${\cal{P}}'$ is updated accordingly. After the type $2$ admissible forwarding move, $x'$ is labeled as unexplored and set as type $2$ node. Then MDFS makes ${\cal{N}}(x')$ with ${\cal{L(N}}(x'))={\cal{L}}(N)$ the next super node to be explored when ${\cal{E}}_u^{{\cal{L}}(N)}$ is updated as ${\cal{E}}_u^{{\cal{L}}(N)}+(x,y_1)-(x_1,y_1)+(x_1,y_2)-(x_2,y_2)+...+(x_{m-1},y_m)-(x',y_m)$ and $|{\cal{E}}_u^{{\cal{L}}(N)}|$ keeps to be $|{\cal{P}}|$.
\item
Type $3$ admissible forwarding move: A type $3$ admissible forwarding move can be understood as that MDFS updates ${\cal{P}}'$ along the edge $(x,y)$ as follows. Let $P_x$ be the path in ${\cal{P}}'$ existing before the current move where $(x,y)$ belongs. If $P_x\neq P_{|{\cal{P}}|+1}$, then after the current move, MDFS updates $P_x$ and $P_{|{\cal{P}}|+1}$ to $P_x'$ and $P_{|{\cal{P}}|+1}'$ as follows. $P_x'=P_{|{\cal{P}}|+1}(S,N)+P_x(N,D)$ and $P_{|{\cal{P}}|+1}'=P_x(S,{\cal{N}}(x))$. If $P_x=P_{|{\cal{P}}|+1}$, then after the current move, MDFS updates $P_{|{\cal{P}}|+1}$ to $P_{|{\cal{P}}|+1}'$ as follows. $P_{|{\cal{P}}|+1}'=P_{|{\cal{P}}|+1}(S,{\cal{N}}(x))$. ${\cal{P}}'$ is updated accordingly. After the type $3$ move, $y$ is labeled explored, $x$ is labeled as unexplored and set as type $3$ node. Then MDFS makes ${\cal{N}}(x)$ with ${\cal{L(N}}(x))={\cal{L}}(N)-1$ the next super node to be explored when ${\cal{E}}_u^{{\cal{L}}(N)-1}$ is updated as ${\cal{E}}_u^{{\cal{L}}(N)-1}-(x,y)$ and $|{\cal{E}}_u^{{\cal{L}}(N)-1}|$ decreases from $|{\cal{P}}|+1$ to $|{\cal{P}}|$.
\end{enumerate}
\end{enumerate}
\end{definition}

\begin{lemma}\label{lemma7}
In iteration $k+1$ of our algorithm, MDFS defined in Definition \ref{definition2} maintains a set of size $k$ linearly independent $S-D$ paths while it tries to complete a $(k+1)$th $S-D$ path $P_{k+1}$. And when MDFS explores a super node $N$, it means MDFS extends $P_{k+1}$ from $S$ to $N$ in addition to $k$ complete $S-D$ paths (all stored in ${\cal{P}}'$) and there are $k+1$ linearly independent used path edges in each of the first ${\cal{L}}(N)-1$ layer cuts and $k$ linearly independent used path edges in each of the rest layer cuts. These two facts lead to the conclusion that when MDFS returns True in iteration $k+1$, totally $k+1$ linearly independent $S-D$ paths are found (and stored in ${\cal{P}}'$).
\end{lemma}
\begin{proof}
We prove by induction. Assume that in the first $k$ iterations of our algorithm, MDFS already finds $k$ linearly independent $S-D$ paths. Then at the beginning of iteration $k+1$ before we call MDFS, these two statements are clearly true. Now it's sufficient for us to show that these three kinds of forwarding moves allowed in MDFS keep these two statements true after each forwarding move of MDFS. First, clearly a type $1$ forwarding move along an edge $(x,y)$ increases $|{\cal{E}}_u^{{\cal{L}}(x)}|$ from $k$ to $k+1$ with rank$(T({\cal{E}}_u^{{\cal{L}}(x)}))=k+1$ and it has no impact to the set of $k$ $S-D$ paths existing before the current move. So these two statements remain true after a type $1$ forwarding move. Second, a type $2$ forwarding move keeps $|{\cal{E}}_u^{{\cal{L}}(x)}|$ unchanged with rank$(T({\cal{E}}_u^{{\cal{L}}(x)}))=k$ (according to Lemma \ref{lemma1}) while it updates the paths $P_i,1\leq i\leq m$ and $P_{k+1}$ to $P_i',1\leq i\leq m$ and $P_{k+1}'$. Clearly the rest $k-m$ $S-D$ paths existing before the current move are not impacted by the move. So these two statements remain true after a type $2$ forwarding move. Third, a type $3$ forwarding move extends $P_{k+1}$ along an edge $(x,y)$ from super node ${\cal{N}}(y)$ to ${\cal{N}}(x)$. Because MDFS reaches ${\cal{N}}(y)$ before the current move, $|{\cal{E}}_u^{{\cal{L(N}}(x)}|=k+1$ and after the current move, $|{\cal{E}}_u^{{\cal{L(N}}(x)}|=k$ with $(x,y)$ deleted. Clearly MDFS maintains $k$ linearly independent $S-D$ paths after the current move when it proceeds to explore ${\cal{N}}(x)$. So these two statements remain true after a type $3$ forwarding move.
\end{proof}

\begin{lemma}\label{lemm9}
Let ${\cal{P}}$ be the set of $S-D$ paths found by our algorithm and $|{\cal{P}}|=K$, then the paths in ${\cal{P}}$ correspond to some transmission scheme of rate $K$ from $S$ to $D$ for the underlying deterministic network.
\end{lemma}
\begin{proof}
We prove by constructing a transmission scheme of rate $K$ from $S$ to $D$ for the underlying deterministic network by using the $K$ paths in ${\cal{P}}$. Lemma \ref{lemma7} says that these $K$ paths in ${\cal{P}}$ found by our algorithm are linearly independent. Let ${\cal{V}}_{y,N}^{{\cal{P}}}\subseteq {\cal{V}}_{y,N}$ (${\cal{V}}_{x,N}^{{\cal{P}}}\subseteq {\cal{V}}_{x,N}$) be the set of transmitting (receiving) nodes incident to $N$ that are used by some paths in ${\cal{P}}$. Clearly $|{\cal{V}}_{y,N}^{{\cal{P}}}|=|{\cal{V}}_{x,N}^{{\cal{P}}}|$. The relay function at each relay super node $N$ could be any one-one mapping between ${\cal{V}}_{y,N}^{{\cal{P}}}$ and ${\cal{V}}_{x,N}^{{\cal{P}}}$, i.e., each received bit from ${\cal{V}}_{y,N}^{{\cal{P}}}$ is transmitted forward by a unique transmitting node in ${\cal{V}}_{x,N}^{{\cal{P}}}$ specified by the mapping. Clearly each such mapping corresponds to a full rank adjacency matrix. For simplicity, we can treat those relay functions as intra-layer paths which together with the paths in ${\cal{P}}$ (treated as inter-layer paths) completely specify a transmission scheme from $S$ to $D$. Since in each of the $L-1$ inter-layers and each of the $L-2$ intra-layers the adjacency matrix for used path edges has full rank $K$, the overall transfer matrix from $S$ to $D$ along these paths is the product of all these adjacency matrix and also has full rank $K$ which guarantees that transmission information rate $K$ is allowed between $S-D$ with the transmission scheme defined by these $K$ paths in ${\cal{P}}$ and any one-one mapping function for each relay super node.
\end{proof}

Let's justify the rule that type $2$ nodes shouldn't start any type $2$ forwarding moves in MDFS. By definition, a type $1$ or type $3$ transmitting node $x\in{\cal{V}}_{N}^{ad}$ can start type $1$ or type $2$ admissible forwarding moves of $N$ while a type $2$ transmitting node $x'\in{\cal{V}}_{N}^{ad}$ can only start type $1$ admissible forward moves of $N$. Lemma \ref{lemma1} helps in explaining why type $2$ transmitting node shouldn't start any type $2$ forwarding moves. When $x'$ is labeled as type $2$ node after a type $2$ forwarding move along the path $P_{x\rightarrow x'}$ for some type $1$ or type $3$ node $x$, it means that $x'\in{\cal{V}}_{xspan}^{x}$. It's easy to see that when $x'$ is explored right after the type $2$ forwarding move along path $P_{x\rightarrow x'}$, ${\cal{V}}_{xspan}^{x'}={\cal{V}}_{xspan}^{x}+x-x'$. Since $x$ is explored and $x$ is of type $1$ or type $3$, all transmitting nodes in ${\cal{V}}_{xspan}^{x}$ could be explored following type $2$ forwarding moves from $x$. It means that all transmitting nodes in ${\cal{V}}_{xspan}^{x'}$ are guaranteed to have the chance of being explored by MDFS as long as $x$ is explored and all type $2$ forwarding moves from $x$ are allowed. So type $2$ node $x'$ shouldn't start any type $2$ forwarding moves to avoid redundancy.

\subsection{Implementation Details of MDFS}\label{implement}

The MDFS algorithm is implemented in the (Res,${\cal{P}}')=$MDFS$({\cal{G,P,P'}},N,D,k)$ function. The input parameters are the graph ${\cal{G}}$, information about the $k-1$ found $S-D$ paths stored in structure ${\cal{P}}$, information about the updated $k-1$ $S-D$ paths and partial path $P_{k}$ stored in a structure ${\cal{P}}'$, starting super node $N$ from where to complete $P_{k}$, ending super node $D$ where to end $P_{k}$ and the iteration number $k$. In each iteration, we call MDFS by initializing $N$ to be $S$ for finding an $S-D$ path. The function returns (True$, {\cal{P}}')$ with $|{\cal{P}}'|=|{\cal{P}}|+1=k$ if the $k$th $S-D$ path is found in current iteration and (False$, {\cal{P}}')$ with ${\cal{P}}'={\cal{P}}$ if no path is found. Our algorithm calls MDFS$({\cal{G}},{\cal{P}},{\cal{P}}',S,D,k)$ in iteration $k$ for increasing $k$ until MDFS$({\cal{G}},{\cal{P}},{\cal{P}}',S,D,K+1)$ returns Res=False for some $K$, then our algorithm stops and claims that $K=C$. Each time MDFS(${\cal{G}},{\cal{P}},{\cal{P}}',S,D,k$) returns Res=True, ${\cal{P}}$ is updated as ${\cal{P}}'$ and used in next iteration.

\subsection{Accelerating The MDFS Algorithm}\label{accelerate}

From definition, we see that the main computational complexity of the MDFS algorithm comes from the rank computation of the associated binary adjacency matrix in deciding whether an edge $(x,y)$ qualifies for a type $1$ admissible forwarding move when $x$ of type $1$, type $2$ or type $3$ is explored and from the computation of ${\cal{V}}_{xspan}^{x}$ in deciding whether a type $2$ admissible forwarding move exists between a pair $x,x'$ when $x$ of type $1$ or type $3$ is explored by MDFS. In this subsection we explore some useful combinatorial properties related to the rank of the binary adjacency matrix and develop an equivalent but computationally simple method to replace the rank computation in MDFS.

In iteration $k+1$ of our algorithm, let $Y_{x_j}^{k+1}=\{y: {\cal{L}}(y)={\cal{L}}(x_j)+1=i+1 {\mbox{ and }} y\not\in{\cal{V}}_{yu,x_j}^{i}\}$, let $Y_{x_j,c}^{k+1}=\{y: y\in Y_{x_j}^{k+1}$ and $(x,y)\in{\cal{E}}$ for some $x\in\{x_j\cup{\cal{V}}_{xspan}^{x_j}\}\}$ and let $Y_{x_j,fr}^{k+1}=\{y: y\in Y_{x_j}^{k+1}$ and rank$(T({\cal{V}}_{xu,x_j}^{i}+x_j,{\cal{V}}_{yu,x_j}^{i}+y))=k+1\}$.

\begin{lemma}\label{lemma2}
$Y_{x_j,fr}^{k+1}$ is a subset of $Y_{x_j,c}^{k+1}$. The rank check of the adjacency matrix (rank$(T({\cal{V}}_{xu,x_j}^{i}+x_j,{\cal{V}}_{yu,x_j}^{i}+y))=k$ or $k+1$) for $\forall y\in Y_{x_j,c}^{k+1}$ is equivalent to checking $T(x_j,y)=\sum_{x\in{\cal{V}}_{xspan}^{x_j}}T(x,y)$ or not, whose computational complexity is bounded by $O(k)$.
\end{lemma}
\begin{proof}
First we prove that for $\forall y\in Y_{x_j}^{k+1}$, if there is no $(x,y)\in{\cal{E}}$ with $x\in\{x_j\cup{\cal{V}}_{xspan}^{x_j}\}$, then $T({\cal{V}}_{xu,x_j}^{i}+x_j,{\cal{V}}_{yu,x_j}^{i}+y)=k$. If there is no $(x,y)\in{\cal{E}}$ with $x\in\{x_j\cup{\cal{V}}_{xspan}^{x_j}\}$, then the vector $T({\cal{V}}_{xu,x_j}^{i}+x_j,y)$ has zero entries on rows indexed by $x_j\cup{\cal{V}}_{xspan}^{x_j}$, together with equation (\ref{eqn21}) we have $T(x_j,{\cal{V}}_{yu,x_j}^{i}+y)=\sum_{x\in{\cal{V}}_{xspan}^{x_j}}T(x,{\cal{V}}_{yu,x_j}^{i}+y)$. So rank$(T({\cal{V}}_{xu,x_j}^{i}+x_j,{\cal{V}}_{yu,x_j}^{i}+y))=$ rank$(T({\cal{V}}_{xu,x_j}^{i},{\cal{V}}_{yu,x_j}^{i}+y))=k$. So we proved $Y_{x_j,fr}^{k+1}$ is a subset of $Y_{x_j,c}^{k+1}$.

We know rank$(T({\cal{V}}_{xu,x_j}^{i},{\cal{V}}_{yu,x_j}^{i}+y))=k$ and $|{\cal{V}}_{xu,x_j}^{i}|=k$. If rank$(T({\cal{V}}_{xu,x_j}^{i}+x_j,{\cal{V}}_{yu,x_j}^{i}+y))=k$, then it must be $T(x_j,{\cal{V}}_{yu,x_j}^{i}+y)=\sum_{x\in {\cal{V}}_{xspan}^{x_j'}}T(x,{\cal{V}}_{yu,x_j}^{i}+y)$ for some ${\cal{V}}_{xspan}^{x_j'}\subseteq{\cal{V}}_{xu,x_j}^{i}$. Given equation (\ref{eqn21}) and ${\cal{V}}_{xspan}^{x_j}$ is unique, we have ${\cal{V}}_{xspan}^{x_j'}={\cal{V}}_{xspan}^{x_j}$ and $T(x_j,{\cal{V}}_{yu,x_j}^{i}+y) = \sum_{x\in{\cal{V}}_{xspan}^{x_j'}}T(x,{\cal{V}}_{yu,x_j}^{i}+y)$ is equivalent to $T(x_j,y)=\sum_{x\in{\cal{V}}_{xspan}^{x_j}}T(x,y)$. So rank$(T({\cal{V}}_{xu,x_j}^{i}+x_j,{\cal{V}}_{yu,x_j}^{i}+y))=k$ or $k+1$ is equivalent to $T(x_j,y)=\sum_{x\in{\cal{V}}_{xspan}^{x_j}}T(x,y)$ or not. \end{proof}

Lemma \ref{lemma2} says that given ${\cal{V}}_{xspan}^{x}$ we can simplify the rank computation of the associated binary adjacency matrix in deciding whether an edge $(x,y)$ qualifies for a type $1$ admissible forwarding move when $x$ of type $1$, type $2$ or type $3$ is explored. We also know that when $x'$ is explored as a type $2$ node right after the type $2$ forwarding move along path $P_{x\rightarrow x'}$ for some $x$ of type $1$ or type $3$, ${\cal{V}}_{xspan}^{x'}={\cal{V}}_{xspan}^{x}+x-x'$, so we don't need to compute ${\cal{V}}_{xspan}^{x'}$ for any $x'$ of type $2$ being explored as long as ${\cal{V}}_{xspan}^{x}$ is computed.


\begin{lemma}\label{lemma4}
In the last iteration $K+1$ of our algorithm when no more $S-D$ path is found by MDFS, all ${\cal{N}}(x_k)$ with $x_k\in{\cal{V}}_{xspan}^{x_j}$ and all ${\cal{N}}(y)$ with $y\in Y_{x_j,fr}^{K+1}$ must have been explored by MDFS before MDFS returns if $x_j$ has ever been explored.
\end{lemma}
\begin{proof}
Since there is no more $S-D$ path exists in the last iteration of our algorithm, all possible admissible moves would be tried by MDFS. It means each ${\cal{N}}(x_k)$ with $x_k\in{\cal{V}}_{xspan}^{x_j}$ would be explored by following a type $2$ forwarding move along the path $P_{x_j\rightarrow x_k}$ if $x_j$ of type $1$ or type $3$ has ever been explored. For any $x_i$ of type $2$ being explored right after a type $2$ move along the path $P_{x_j\rightarrow x_i}$ starting from some $x_j$ of type $1$ or type $3$, we know that $x_j$ is explored and ${\cal{V}}_{xspan}^{x_i}={\cal{V}}_{xspan}^{x_j}+x_j-x_i$. So if all ${\cal{N}}(x_k)$ with $x_k\in{\cal{V}}_{xspan}^{x_j}$ are explored, all ${\cal{N}}(x_k')$ with $x_k'\in{\cal{V}}_{xspan}^{x_i}$ are also explored.

From Lemma \ref{lemma2}, for $\forall y\in Y_{x_j,fr}^{K+1}$, there is $(x,y)\in{\cal{E}}$ for some $x\in\{x_j\cup{\cal{V}}_{xspan}^{x_j}\}$. We already proved that all ${\cal{N}}(x_k)$ with $x_k\in{\cal{V}}_{xspan}^{x_j}$ would be explored by MDFS in the last iteration of our algorithm if $x_j$ has ever been explored. So the edge $(x,y)$ with $x\in\{x_j\cup{\cal{V}}_{xspan}^{x_j}\}$ and $y\in Y_{x_j,fr}^{K+1}$ must be considered by MDFS in the last iteration given that $x_j$ has ever been explored. If $x=x_j$, $(x,y)$ allows a type $1$ forwarding move when $x_j$ is explored and then ${\cal{N}}(y)$ will be explored. Now assume $x\in{\cal{V}}_{xspan}^{x_j}$. Given $x_j$ is explored, $x$ is also explored and ${\cal{V}}_{xu,x}^{i}+x={\cal{V}}_{xu,x_j}^{i}+x_j$ and ${\cal{V}}_{yu,x}^{i}={\cal{V}}_{yu,x_j}^{i}$. Given $y\in Y_{x_j,fr}^{K+1}$, i.e., rank$(T({\cal{V}}_{xu,x_j}^{i}+x_j,{\cal{V}}_{yu,x_j}^{i}+y))=K+1$,
we have rank$(T({\cal{V}}_{xu,x}^{i}+x,{\cal{V}}_{yu,x}^{i}+y))=K+1$. Then $(x,y)$ allows a type $1$ forwarding move when $x$ is explored and then ${\cal{N}}(y)$ will be explored. So all ${\cal{N}}(y)$ with $y\in Y_{x_j,fr}^{K+1}$ must be explored by MDFS in the last iteration of our algorithm if $x_j$ has ever been explored.
\end{proof}

Table \ref{table1} gives a pseudo-code description of the MDFS algorithm described above. Refer to \cite{myhomepage} for an implementation of the algorithm proposed in the current paper.

\begin{table}[htbp]
\caption{}
\begin{center}
\begin{tabular}{l}
\hline
\hline
(Res=\{True, False\},${\cal{P}}'$)=\textbf{MDFS}(${\cal{G}},{\cal{P}},{\cal{P}}',N,D,k$)\\
$\left\{\begin{array}{l}
{\mbox{SetLabel}}(N,{\mbox{explored}})\\
{\mbox{for }}\forall x\in{\cal{V}}_{x,N}{\mbox{ with }}x\not\in{\cal{E}}_{xu}^{{\cal{L}}(N)}{\mbox{ and GetLabel}}(x)={\mbox{unexplored and GetType}}(x)=2\\
\left\{\begin{array}{l}
{\mbox{SetLabel}}(x,{\mbox{explored}})\\
{\mbox{for }}\forall e=(x,y)\in {\cal{G}}.{\mbox{incidentEdgeType1}}(x){\mbox{ and GetLabel}}({\cal{N}}(y))={\mbox{unexplored}}\\
\left\{\begin{array}{l}
{\mbox{if }}T(x,y)\neq\sum_{x'\in{\cal{V}}_{xspan}^{x}}T(x',y)\\
\left\{\begin{array}{l}
{\cal{E}}^{{\cal{L}}(N)}_u\leftarrow{\cal{E}}^{{\cal{L}}(N)}_u+e; {\mbox{ Update}}({\cal{P}}')\\
{\mbox{if }}{\cal{N}}(y)=D{\mbox{ }}\{{\mbox{ return (True,}}{\cal{P}}'){\mbox{ }}\}\\
{\mbox{else}}
\left\{\begin{array}{l}
({\mbox{Res}},{\cal{P}}')={\mbox{MDFS}}({\cal{G}},{\cal{P}},{\cal{P}}',{\cal{N}}(y),D,k)\\
{\mbox{if Res}}={\mbox{True }}\{{\mbox{ return (True,}}{\cal{P}}'){\mbox{ }}\}\\
 \end{array}\right.\\
{\cal{E}}^{{\cal{L}}(N)}_u\leftarrow{\cal{E}}^{{\cal{L}}(N)}_u-e; {\mbox{ Restore}}({\cal{P}}')\\
\end{array}\right.\\
\end{array}\right.\\
\end{array}\right.\\
{\mbox{for }}\forall x\in{\cal{V}}_{x,N}{\mbox{ with }}x\not\in{\cal{E}}_{xu}^{{\cal{L}}(N)}{\mbox{ and GetLabel}}(x)={\mbox{unexplored and GetType}}(x)\neq2\\
{\mbox{(first for $x$ with GetType}}(x)=3{\mbox{, then for $x$ with GetType}}(x)=1)\\
\left\{\begin{array}{l}
{\mbox{SetLabel}}(x,{\mbox{explored}})\\
{\cal{V}}_{xspan}^{x}={\mbox{Span}}(x)\\
{\mbox{for }}\forall e=(x,y)\in {\cal{G}}.{\mbox{incidentEdgeType1}}(x){\mbox{ and GetLabel}}({\cal{N}}(y))={\mbox{unexplored}}\\
\left\{\begin{array}{l}
{\mbox{if }}T(x,y)\neq\sum_{x'\in{\cal{V}}_{xspan}^{x}}T(x',y)\\
\left\{\begin{array}{l}
{\cal{E}}^{{\cal{L}}(N)}_u\leftarrow{\cal{E}}^{{\cal{L}}(N)}_u+e; {\mbox{ Update}}({\cal{P}}')\\
{\mbox{if }}{\cal{N}}(y)=D{\mbox{ }}\{{\mbox{ return (True,}}{\cal{P}}'){\mbox{ }}\}\\
{\mbox{else}}
\left\{\begin{array}{l}
({\mbox{Res}},{\cal{P}}')={\mbox{MDFS}}({\cal{G}},{\cal{P}},{\cal{P}}',{\cal{N}}(y),D,k)\\
{\mbox{if Res}}={\mbox{True }}\{{\mbox{ return (True,}}{\cal{P}}'){\mbox{ }}\}\\
 \end{array}\right.\\
{\cal{E}}^{{\cal{L}}(N)}_u\leftarrow{\cal{E}}^{{\cal{L}}(N)}_u-e; {\mbox{ Restore}}({\cal{P}}')\\
\end{array}\right.\\
\end{array}\right.\\
P_{x\rightarrow{\cal{V}}_{xspan}^{x}}={\mbox{FindIndPaths}}(x,{\cal{V}}_{xspan}^{x})\\
{\mbox{for }}\forall x'\in{\cal{V}}_{xspan}^{x}{\mbox{ with }}P_{x\rightarrow x'}=\{e_1,e_2,...e_{2m}\}=\{(x,y_1),(y_1,x_1),(x_1,y_2),...(y_m,x_m)=(y_m,x')\}\\
\left\{\begin{array}{l}
{\mbox{SetLabel}}(x', {\mbox{unexplored); SetType}}(x',2); {\cal{V}}_{xspan}^{x'}={\cal{V}}_{xspan}^{x}-x'+x\\
{\cal{E}}^{{\cal{L}}(N)}_u\leftarrow{\cal{E}}^{{\cal{L}}(N)}_u+e_1-e_2+e_3-e_4+...+e_{2m-1}-e_{2m};{\mbox{ Update}}({\cal{P}}')\\
({\mbox{Res}},{\cal{P}}')={\mbox{MDFS}}({\cal{G}},{\cal{P}},{\cal{P}}',{\cal{N}}(x'),D,k)\\
{\mbox{if Res}}={\mbox{True }}\{{\mbox{ return (True,}}{\cal{P}}'){\mbox{ }}\}\\
{\cal{E}}^{{\cal{L}}(N)}_u\leftarrow{\cal{E}}^{{\cal{L}}(N)}_u-e_1+e_2-e_3+e_4-...-e_{2m-1}+e_{2m};{\mbox{ Restore}}({\cal{P}}')\\
\end{array}\right.\\
\end{array}\right.\\
{\mbox{for }}\forall e=(x,y)\in{\cal{G}}.{\mbox{incidentEdgeType3}}(N){\mbox{ with GetLabel}}(y)={\mbox{unexplored}}\\
\left\{\begin{array}{l}
{\mbox{SetLabel}}(y,{\mbox{explored}})\\
{\mbox{SetLabel}}(x,{\mbox{unexplored); SetType}}(x,3)\\
{\cal{E}}^{{\cal{L}}(N)-1}_u\leftarrow{\cal{E}}^{{\cal{L}}(N)-1}_u-e;{\mbox{ Update}}({\cal{P}}')\\
({\mbox{Res}},{\cal{P}}')={\mbox{MDFS}}({\cal{G}},{\cal{P}},{\cal{P}}',{\cal{N}}(x),D,k)\\
{\mbox{if Res}}={\mbox{True }}\{{\mbox{ return (True,}}{\cal{P}}'){\mbox{ }}\}\\
{\cal{E}}^{{\cal{L}}(N)-1}_u\leftarrow{\cal{E}}^{{\cal{L}}(N)-1}_u+e;{\mbox{ Restore}}({\cal{P}}')\\
\end{array}\right.\\
{\mbox{return (False,}}{\cal{P}}')\\
\end{array}\right.$\\
\hline
\hline
\end{tabular}
\end{center}
\label{table1}
\end{table}

\section{Algorithm Analysis\label{analysis}}

\subsection{Complexity Analysis}\label{complexity}

\begin{theorem}\label{theorem2}
MDFS algorithm defined in Section \ref{ouralgo} terminates in finite time as the total number of explorations to transmitting nodes invoked by MDFS is bounded by $O(|{\cal{V}}_x|\cdot k)$ in iteration $k$ of our algorithm. The total computational complexity of our algorithm is bounded by $O(|{\cal{V}}_x|\cdot C^4+d\cdot |{\cal{V}}_x|\cdot C^3)$ if our algorithm stops after finding $C$ linearly independent $S-D$ paths.
\end{theorem}
\begin{proof}
In MDFS, a transmitting node is explored or re-explored exclusively as an unexplored type $1$, type $2$ or type $3$ node. Let $k_1$, $k_2$ and $k_3$ be the total number of labeling of a transmitting node as unexplored type $1$, type $2$ and type $3$ node respectively. To prove MDFS algorithm terminates, we can equivalently prove that $k_1$, $k_2$ and $k_3$ are finite.
\begin{claim}
$k_1\leq|{\cal{V}}_x|$.
\end{claim}
\begin{proof}
All nodes in ${\cal{V}}_x$ are labeled as unexplored type $1$ nodes at the beginning of the iteration. After an unexplored type $1$ node is explored, it's labeled as explored and never relabeled as unexplored type $1$ node again. So $k_1\leq|{\cal{V}}_x|$.
\end{proof}

\begin{claim}
$k_3\leq|{\cal{V}}_x|$.
\end{claim}
\begin{proof}
A node is labeled as unexplored type $3$ node only after a type $3$ admissible forwarding move. We next show that the number of type $3$ admissible forwarding moves is bounded by $|{\cal{V}}_x|$. A type $3$ admissible forwarding move only happens on an edge $(x,y)$ with $y\in{\cal{V}}_{y{\cal{P}}}$ unexplored and y is labeled as explored after a type $3$ move and never relabeled as unexplored again. Since $|{\cal{V}}_{y{\cal{P}}}|=|{\cal{V}}_{x{\cal{P}}}|\leq|{\cal{V}}_x|$, so the total number of type $3$ moves is bounded by $|{\cal{V}}_x|$. The total number of labeling of transmitting nodes as unexplored type $3$ nodes is therefore bounded by $|{\cal{V}}_x|$.
\end{proof}

\begin{claim}
$k_2\leq k\cdot(k_1+k_3)$.
\end{claim}
\begin{proof}
A node is labeled as unexplored type $2$ node only after a type $2$ admissible forwarding move. We next show that the total number of type $2$ moves is bounded by $k_2\leq k\cdot(k_1+k_3)$. A type $2$ admissible forwarding move only starts from some type $1$ or type $3$ node being explored. The total number of type $1$ or type $3$ nodes explored by MDFS is bounded by $2|{\cal{V}}_x|$. When a type $1$ or a type $3$ node $x$ is explored, the total number of type $2$ admissible forwarding moves that it starts is bounded by $|{\cal{V}}_{xspan}^x|\leq k$. So the total number of type $2$ moves is bounded by $2\cdot k\cdot|{\cal{V}}_x|$. The total number of labeling of transmitting nodes as unexplored type $2$ nodes is therefore bounded by $2\cdot k\cdot|{\cal{V}}_x|$.
\end{proof}

Now it is easy to conclude that the total number of explorations to transmitting nodes invoked by MDFS is bounded by $O(|{\cal{V}}_x|\cdot k)$ in iteration $k$ of our algorithm, so MDFS terminates in finite time.

Now let's consider the computational complexity of our algorithm. We already proved that in iteration $k$ the total number of type $1$, type $2$ and type $3$ transmitting nodes explored by MDFS is bounded by $|{\cal{V}}_x|$, $2|{\cal{V}}_x|\cdot k$ and $|{\cal{V}}_x|$ respectively. The worst case in computation for all these transmitting nodes is that (1) MDFS computes ${\cal{V}}_{xspan}^{x}$ (with complexity bounded by $O(k^3)$ based on Lemma \ref{lemma1}) and finds the paths $P_{x\rightarrow{\cal{V}}_{xspan}^x}$ (with complexity bounded by $O(k^2)$ based on Lemma \ref{lemma5}) for any $x$ being explored as type $1$ or type $3$ node and (2) MDFS checks the rank of the binary adjacency matrix associated with each incident edge to $x$ (with complexity bounded by $O(k)$ based on Lemma \ref{lemma2}) for any $x$ being explored as type $1$, type $2$ or type $3$ node. Clearly, the total number of such edges is bounded by $2dk|{\cal{V}}_x|$ in iteration $k$. So the computational complexity of MDFS in iteration $k$ is bounded by $O(|{\cal{V}}_x|\cdot k^3+d\cdot|{\cal{V}}_x|\cdot k^2)$. The total computational complexity of our algorithm is bounded by $O(|{\cal{V}}_x|\cdot C^4+d\cdot |{\cal{V}}_x|\cdot C^3)$ if our algorithm stops after finding $C$ linearly independent $S-D$ paths.
\end{proof}

Our algorithm shows a significant improvement in terms of computational complexity over the algorithms for solving the same problem in \cite{aurore2009_combinatorial_algo_deterministic} by Amaudruz and Fragouli (whose complexity is bounded by $O(M\cdot |{\cal{E}}|\cdot C^5)$) and over the algorithm in \cite{sadegh2009_combinatorialstudyofdeterministic} by Yazdi and Savari (whose complexity is bounded by $O(L^8\cdot M^{12}\cdot h_0^3+L\cdot M^6\cdot C\cdot h_0^4)$). Note that here $M\geq d$ (because each transmitting node can transmit at most one bit information to each super node by definition of the deterministic channel model), $|{\cal{E}}|\geq|{\cal{V}}_x|$ (because of broadcasting) and $h_0\geq C$ (because $C$ cannot be larger than the maximum number of transmitting nodes among all layers).

\subsection{Proof of Correctness\label{proof}}

\begin{theorem}\label{theorem1}
Our algorithm finds $C$ linearly independent paths in a linear layered deterministic relay network ${\cal{G}}$ where $C$ is the unicast capacity (or the minimum cut value among all cuts separating the source from the destination) of ${\cal{G}}$.
\end{theorem}
\begin{proof}
We prove Theorem \ref{theorem1} by proving that when our algorithm stops the number of paths we find equals some cut value in ${\cal{G}}$.

If MDFS returns Res=False in iteration $K+1$ of our algorithm, then it finds the maximum number $K$ of linearly independent $S-D$ paths in ${\cal{G}}$ and we claim $K=C$. Consider the cut $\Omega_K$ separating the super nodes labeled explored from the super nodes labeled unexplored in iteration $K+1$ when the algorithm stops so that $S\in\Omega_K$.
Clearly $\Omega_K$ is a cut separating $S$ from $D$ and $S\in\Omega_K$, $D\in\Omega_K^c$. We prove Theorem \ref{theorem1} by proving that this cut value equals $K$, i.e, rank$(T({{\cal{E}}_{\Omega_K}}))=K$.

Let ${\cal{E}}_{{\cal{P}}}^i=\{(x,y):(x,y)\in {\cal{E}}_{{\cal{P}}} {\mbox{ and }} {\cal{L}}(x)=i\}$. Let ${\cal{V}}^i_{x{\cal{P}}}=\{x: (x,y)\in {\cal{E}}_{{\cal{P}}}^i\}$ and ${\cal{V}}^i_{y{\cal{P}}}=\{y: (x,y)\in {\cal{E}}_{{\cal{P}}}^i\}$.
We divide the set ${\cal{E}}_{{\cal{P}}}^i$ into four subgroups: ${\cal{E}}_{{\cal{P}}1}^i=\{(x,y):(x,y)\in{\cal{E}}_{{\cal{P}}}^i,{\cal{N}}(x)\in\Omega_K,{\cal{N}}(y)\in\Omega_K^c\}$, ${\cal{E}}_{{\cal{P}}2}^i=\{(x,y):(x,y)\in{\cal{E}}_{{\cal{P}}}^i,{\cal{N}}(x)\in\Omega_K,{\cal{N}}(y)\in\Omega_K\}$, ${\cal{E}}_{{\cal{P}}3}^i=\{(x,y):(x,y)\in{\cal{E}}_{{\cal{P}}}^i,{\cal{N}}(x)\in\Omega_K^c,{\cal{N}}(y)\in\Omega_K^c\}$ and ${\cal{E}}_{{\cal{P}}4}^i=\{(x,y):(x,y)\in{\cal{E}}_{{\cal{P}}}^i,{\cal{N}}(x)\in\Omega_K^c,{\cal{N}}(y)\in\Omega_K\}$. We divide the sets ${\cal{V}}_{x{\cal{P}}}^i$ (${\cal{V}}_{y{\cal{P}}}^i$) into four subgroups accordingly, ${\cal{V}}_{x{\cal{P}}j}^i$ (${\cal{V}}_{y{\cal{P}}j}^i$), $1\leq j\leq 4$. Clearly, the subgroups ${\cal{V}}_{x{\cal{P}}j}^i,1\leq j\leq 4$ are disjoint, so are true for subgroups ${\cal{V}}_{y{\cal{P}}j}^i,1\leq j\leq 4$. Denote $|{\cal{E}}_{{\cal{P}}1}^i|=|{\cal{V}}_{x{\cal{P}}1}^i|=|{\cal{V}}_{y{\cal{P}}1}^i|=K_i$. Clearly $K_i$ is the number of our identified paths (or path edges) that cross the cut $\Omega_K$ in layer cut $i$. Denote $|{\cal{E}}_{{\cal{P}}2}^i|=K_{i2}$, $|{\cal{E}}_{{\cal{P}}3}^i|=K_{i3}$ and $|{\cal{E}}_{{\cal{P}}4}^i|=K_{i4}$. Clearly, $K_i+K_{i2}+K_{i3}+K_{i4}=K$.

Let ${\cal{E}}^i_{\Omega_K}=\{(x,y): (x,y)\in {\cal{E}}, {\cal{N}}(x)\in\Omega_K, {\cal{N}}(y)\in\Omega_K^c {\mbox{ and }} {\cal{L}}(x)=i\}, 1\leq i < L$, i.e., ${\cal{E}}^i_{\Omega_K}$ is the intersection of the cut $\Omega_K$ and layer cut $i$. Consider the adjacency matrix $T({\cal{E}}_{{\Omega_K}})$ for $\Omega_K$. It is a block diagonal matrix with each block $T^i_{\Omega_K}$ being the adjacency matrix for ${\cal{E}}^i_{\Omega_K}$ and rank$(T({\cal{E}}_{{\Omega_K}}))=\sum_{i=1}^{L-1}$ rank$(T_{\Omega_K}^i)$. Let ${\cal{V}}^i_{x\Omega_K}=\{x: (x,y)\in {\cal{E}}^i_{\Omega_K}\}$ and ${\cal{V}}^i_{y\Omega_K}=\{y: (x,y)\in {\cal{E}}^i_{\Omega_K}\}$. We also divide the set ${\cal{V}}^i_{x\Omega_K}$ into three subgroups: ${\cal{V}}^i_{x\Omega_K1} = {\cal{V}}^i_{x{\cal{P}}1}$, ${\cal{V}}^i_{x\Omega_K2} = {\cal{V}}^i_{x{\cal{P}}2}$ and ${\cal{V}}^i_{x\Omega_K3} = {\cal{V}}^i_{x\Omega_K}-{\cal{V}}^i_{x\Omega_K1}-{\cal{V}}^i_{x\Omega_K2}$. Similarly, we divide the set ${\cal{V}}^i_{y\Omega_K}$ into three subgroups: ${\cal{V}}^i_{y\Omega_K1} = {\cal{V}}^i_{y{\cal{P}}1}$, ${\cal{V}}^i_{y\Omega_K2} = {\cal{V}}^i_{y{\cal{P}}3}$ and ${\cal{V}}^i_{y\Omega_K3} = {\cal{V}}^i_{y\Omega_K}-{\cal{V}}^i_{y\Omega_K1}-{\cal{V}}^i_{y\Omega_K2}$. Clearly, the subgroups ${\cal{V}}^i_{x\Omega_Kj},1\leq j\leq 3$ are disjoint, so are true for subgroups ${\cal{V}}^i_{y\Omega_Kj},1\leq j\leq 3$.

Denote ${\cal{E}}_{u,x_k}^{i}$, ${\cal{V}}_{xu,x_k}^{i}$ and ${\cal{V}}_{yu,x_k}^{i}$ as the instantaneous sets of ${\cal{E}}_u^i$, ${\cal{V}}_{xu}^i$ and ${\cal{V}}_{yu}^i$ when a transmitting node $x_k$ is being explored with $i={\cal{L}}(x_k)$ in iteration $K+1$ of our algorithm. We divide the set ${\cal{E}}_{u,x_k}^{i}$ into four subgroups as follows. Let ${\cal{E}}_{u,x_k,1}^{i}=\{(x,y)\in{\cal{E}}_{u,x_k}^{i},{\cal{N}}(x)\in\Omega_K,{\cal{N}}(y)\in\Omega_K^c\}$, ${\cal{E}}_{u,x_k,2}^{i}=\{(x,y)\in{\cal{E}}_{u,x_k}^{i},{\cal{N}}(x)\in\Omega_K,{\cal{N}}(y)\in\Omega_K\}$, ${\cal{E}}_{u,x_k,3}^{i}=\{(x,y)\in{\cal{E}}_{u,x_k}^{i},{\cal{N}}(x)\in\Omega_K^c,{\cal{N}}(y)\in\Omega_K^c\}$ and ${\cal{E}}_{u,x_k,4}^{i}=\{(x,y)\in{\cal{E}}_{u,x_k}^{i},{\cal{N}}(x)\in\Omega_K^c,{\cal{N}}(y)\in\Omega_K\}$. Divide ${\cal{V}}_{xu,x_k}^{i}$ and ${\cal{V}}_{yu,x_k}^{i}$ respectively into four subgroups ${\cal{V}}_{xu,x_k,j}^{i}$ and ${\cal{V}}_{yu,x_k,j}^{i}$ corresponding to ${\cal{E}}_{u,x_k,j}^{i},1\leq j\leq 4$. Clearly $|{\cal{E}}_{u,x_k,1}^{i}|+|{\cal{E}}_{u,x_k,2}^{i}|+|{\cal{E}}_{u,x_k,3}^{i}|+|{\cal{E}}_{u,x_k,4}^{i}|=K$ and rank$(T({\cal{E}}_{u,x_k}^{i}))=K$ for any $x_k$ ever explored by MDFS in iteration $K+1$ of our algorithm based on Lemma \ref{lemma7}.


In the following of this section, we prove a sequence of lemmas before we finally prove Theorem \ref{theorem1}. Unless otherwise stated, we assume that we are in the last iteration $K+1$ of our algorithm.

\begin{lemma}\label{lemm2}
Let $x_k$ be a transmitting node that has been explored by MDFS in iteration $K+1$ of our algorithm with ${\cal{L}}(x_k)=i$. When $x_k$ is explored, we have ${\cal{V}}_{xspan}^{x_k}\subseteq{\cal{V}}_{xu,x_k,1}^i+{\cal{V}}_{xu,x_k,2}^i$ and
\begin{eqnarray}
&&{\mbox{rank}}(T({\cal{V}}_{xu,x_k,1}^i+{\cal{V}}_{xu,x_k,2}^i+x_k,{\cal{V}}_{yu,x_k}^i))\label{eqn50}\\
&=&{\mbox{rank}}(T({\cal{V}}_{xu,x_k,1}^i+{\cal{V}}_{xu,x_k,2}^i,{\cal{V}}_{yu,x_k}^i))=|{\cal{V}}_{xu,x_k,1}^i|+|{\cal{V}}_{xu,x_k,2}^i|\label{eqn51}
\end{eqnarray}
For any $y\in{\cal{V}}_{y\Omega_K3}^i$, we have
\begin{eqnarray}
&&{\mbox{rank}}(T({\cal{V}}_{xu,x_k,1}^i+{\cal{V}}_{xu,x_k,2}^i+x_k,{\cal{V}}_{yu,x_k}^i+y))\\
&=&{\mbox{rank}}(T({\cal{V}}_{xu,x_k,1}^i+{\cal{V}}_{xu,x_k,2}^i,{\cal{V}}_{yu,x_k}^i+y))=|{\cal{V}}_{xu,x_k,1}^i|+|{\cal{V}}_{xu,x_k,2}^i|
\end{eqnarray}
\end{lemma}
\begin{proof}
Based on Lemma \ref{lemma4}, all ${\cal{N}}(x_j)$ with $x_j\in{\cal{V}}_{xspan}^{x_k}$ will finally be explored by MDFS in iteration $K+1$ if $x_k$ has ever been explored. Since for any $x\in{\cal{V}}_{xu,x_k,3}^i+{\cal{V}}_{xu,x_k,4}^i$, ${\cal{N}}(x)$ is not explored by MDFS, so we have ${\cal{V}}_{xspan}^{x_k}\subseteq{\cal{V}}_{xu,x_k,1}^i+{\cal{V}}_{xu,x_k,2}^i$. By definition of ${\cal{V}}_{xspan}^{x_k}$, it's easy to conclude that (\ref{eqn50})-(\ref{eqn51}) hold. By definition ${\cal{N}}(y)$ is not explored by MDFS for any $y\in{\cal{V}}_{y\Omega_K3}^i$. Based on Lemma \ref{lemma4}, we have rank$(T({\cal{V}}_{xu,x_k}^i+x_k,{\cal{V}}_{yu,x_k}^i+y))=K$. Given rank$(T({\cal{V}}_{xu,x_k}^i,{\cal{V}}_{yu,x_k}^i+y))=K$ and the fact that ${\cal{V}}_{xspan}^{x_k}\subseteq{\cal{V}}_{xu,x_k}^i$ is the unique set satisfying rank$(T({\cal{V}}_{xu,x_k}^i+x_k,{\cal{V}}_{yu,x_k}^i))=$rank$(T({\cal{V}}_{xu,x_k}^i,{\cal{V}}_{yu,x_k}^i))=K$, we conclude rank$(T({\cal{V}}_{xu,x_k,1}^i+{\cal{V}}_{xu,x_k,2}^i+x_k,{\cal{V}}_{yu,x_k}^i+y))=$rank$(T({\cal{V}}_{xu,x_k,1}^i+{\cal{V}}_{xu,x_k,2}^i,{\cal{V}}_{yu,x_k}^i+y))=|{\cal{V}}_{xu,x_k,1}^i|+|{\cal{V}}_{xu,x_k,2}^i|$.
\end{proof}

\begin{lemma}
For any $x_k$ with ${\cal{L}}(x_k)=i$ explored by MDFS in iteration $K+1$ of our algorithm, $|{\cal{E}}_{u,x_k,1}^i|+|{\cal{E}}_{u,x_k,3}^i|=|{\cal{E}}_{{\cal{P}}1}^i|+|{\cal{E}}_{{\cal{P}}3}^i|=K_i+K_{i3}$, $|{\cal{E}}_{u,x_k,3}^i|+|{\cal{E}}_{u,x_k,4}^i|=|{\cal{E}}_{{\cal{P}}3}^i|+|{\cal{E}}_{{\cal{P}}4}^i|=K_{i3}+K_{i4}$ and $|{\cal{E}}_{u,x_k,1}^i|-|{\cal{E}}_{u,x_k,4}^i|=|{\cal{E}}_{{\cal{P}}1}^i|-|{\cal{E}}_{{\cal{P}}4}^i|=K_i-K_{i4}$.
\end{lemma}
\begin{proof}
By definition all the super nodes where the nodes in ${\cal{V}}_{x{\cal{P}}3}^i$, ${\cal{V}}_{x{\cal{P}}4}^i$, ${\cal{V}}_{y{\cal{P}}1}^i$ and ${\cal{V}}_{y{\cal{P}}3}^i$ belong to are labeled unexplored when MDFS returns, so during the running time of MDFS in iteration $K+1$, the nodes in ${\cal{V}}_{x{\cal{P}}3}^i$, ${\cal{V}}_{x{\cal{P}}4}^i$, ${\cal{V}}_{y{\cal{P}}1}^i$ and ${\cal{V}}_{y{\cal{P}}3}^i$ are never explored and ${\cal{V}}_{x{\cal{P}}3}^i\subseteq{\cal{V}}_{xu,x_k}^i$, ${\cal{V}}_{x{\cal{P}}4}^i\subseteq{\cal{V}}_{xu,x_k}^i$, ${\cal{V}}_{y{\cal{P}}1}^i\subseteq{\cal{V}}_{yu,x_k}^i$ and ${\cal{V}}_{y{\cal{P}}3}^i\subseteq{\cal{V}}_{yu,x_k}^i$ always hold for any $x_k$ with ${\cal{L}}(x_k)=i$ explored by MDFS. Any ${\cal{N}}(x)$ (or ${\cal{N}}(y)$) with transmitting node $x$ (or receiving node $y$) used by ${\cal{E}}_u^i$ but not by ${\cal{E}}_{{\cal{P}}}^i$ must be explored by MDFS, which means that for any $x_k$ with ${\cal{L}}(x_k)=i$ explored by MDFS ${\cal{V}}_{x{\cal{P}}34}^i={\cal{V}}_{x{\cal{P}}3}^i+{\cal{V}}_{x{\cal{P}}4}^i$ is the complete set of transmitting nodes in ${\cal{V}}_{xu,x_k}^i$ so that for each transmitting node $x$ in this set, ${\cal{N}}(x)$ is unexplored and ${\cal{V}}_{y{\cal{P}}13}^i={\cal{V}}_{y{\cal{P}}1}^i+{\cal{V}}_{y{\cal{P}}3}^i$ is the complete set of receiving nodes in ${\cal{V}}_{yu,x_k}^i$ so that for each receiving node $y$ in this set, ${\cal{N}}(y)$ is unexplored. So $|{\cal{E}}_{u,x_k,1}^i|+|{\cal{E}}_{u,x_k,3}^i|=|{\cal{E}}_{{\cal{P}}1}^i|+|{\cal{E}}_{{\cal{P}}3}^i|=K_i+K_{i3}$ and $|{\cal{E}}_{u,x_k,3}^i|+|{\cal{E}}_{u,x_k,4}^i|=|{\cal{E}}_{{\cal{P}}3}^i|+|{\cal{E}}_{{\cal{P}}4}^i|=K_{i3}+K_{i4}$. $|{\cal{E}}_{u,x_k,1}^i|-|{\cal{E}}_{u,x_k,4}^i|=|{\cal{E}}_{{\cal{P}}1}^i|-|{\cal{E}}_{{\cal{P}}4}^i|=K_i-K_{i4}$ is a straightforward result by subtracting one from the other.
\end{proof}

\begin{lemma}\label{lemm1}
rank$(T({\cal{V}}_{xu,x_k,1}^i+{\cal{V}}_{xu,x_k,2}^i,{\cal{V}}_{yu,x_k,1}^i+{\cal{V}}_{yu,x_k,3}^i))\geq|{\cal{V}}_{xu,x_k,1}^i|-|{\cal{V}}_{xu,x_k,4}^i|=K_i-K_{i4}$ for $x_k$ being explored by MDFS.
\end{lemma}
\begin{proof}
If $K_i-K_{i4}<0$, the statement is obviously true. Assume $K_i-K_{i4}\geq0$. We know rank$(T({\cal{E}}_{u,x_k}^i))=K=|{\cal{E}}_{u,x_k,1}^i|+|{\cal{E}}_{u,x_k,2}^i|+|{\cal{E}}_{u,x_k,3}^i|+|{\cal{E}}_{u,x_k,4}^i|=$ rank$(T({\cal{V}}_{xu,x_k,1}^i+{\cal{V}}_{xu,x_k,2}^i+{\cal{V}}_{xu,x_k,3}^i+{\cal{V}}_{xu,x_k,4}^i,{\cal{V}}_{yu,x_k,1}^i+{\cal{V}}_{yu,x_k,2}^i+{\cal{V}}_{yu,x_k,3}^i+{\cal{V}}_{yu,x_k,4}^i))$. If rank$(T({\cal{V}}_{xu,x_k,1}^i+{\cal{V}}_{xu,x_k,2}^i,{\cal{V}}_{yu,x_k,1}^i+{\cal{V}}_{yu,x_k,3}^i))<|{\cal{V}}_{xu,x_k,1}^i|-|{\cal{V}}_{xu,x_k,4}^i|$, we will have rank$(T({\cal{V}}_{xu,x_k,1}^i+{\cal{V}}_{xu,x_k,2}^i+{\cal{V}}_{xu,x_k,3}^i+{\cal{V}}_{xu,x_k,4}^i,{\cal{V}}_{yu,x_k,1}^i+{\cal{V}}_{yu,x_k,2}^i+{\cal{V}}_{yu,x_k,3}^i+{\cal{V}}_{yu,x_k,4}^i))<|{\cal{E}}_{u,x_k,1}^i|+|{\cal{E}}_{u,x_k,2}^i|+|{\cal{E}}_{u,x_k,3}^i|+|{\cal{E}}_{u,x_k,4}^i|=K$.
\end{proof}

Let $x_k$ be some transmitting node explored by MDFS in iteration $K+1$ of our algorithm. Let $y'\in{\cal{V}}_{yu,x_k,2}^i+{\cal{V}}_{yu,x_k,4}^i$. By definition, ${\cal{N}}(y')$ is explored by MDFS sometime. If $y'\in{\cal{V}}_{y{\cal{P}}}^i$, then $y'$ must have been deleted from ${\cal{V}}_{yu}^i$ in a type $3$ forwarding move of MDFS sometime in iteration $K+1$ of our algorithm given that ${\cal{N}}(y')$ is explored and have already been added back to ${\cal{V}}_{yu,x_k}^i$ if this type $3$ forwarding move on $y'$ happens before the current exploration of $x_k$ otherwise it won't appear in ${\cal{V}}_{yu,x_k}^i$. If $y'\not\in{\cal{V}}_{y{\cal{P}}}^i$, then $y'$ must have been added to ${\cal{V}}_{yu,x_k}^i$ in a type $1$ forwarding move by MDFS before the current exploration to $x_k$ otherwise $y'$ won't appear in ${\cal{V}}_{yu,x_k}^i$ when $x_k$ is explored. It means for each $y'\in{\cal{V}}_{yu,x_k,2}^i+{\cal{V}}_{yu,x_k,4}^i$, $y'$ is either added to ${\cal{V}}_{yu,x_k}^i$ before the current exploration of $x_k$ or is deleted from ${\cal{V}}_{yu,x_k}^i$ after the current exploration of $x_k$.

\begin{lemma}\label{lemm3}
If rank$(T({\cal{V}}_{xu,x_k,1}^i+{\cal{V}}_{xu,x_k,2}^i,{\cal{V}}_{yu,x_k,1}^i+{\cal{V}}_{yu,x_k,3}^i))>|{\cal{V}}_{xu,x_k,1}^i|-|{\cal{V}}_{xu,x_k,4}^i|=K_i-K_{i4}$ for $x_k$ being explored by MDFS, then there exists some nonempty set ${\cal{V}}_{y24,x_k}\subseteq{\cal{V}}_{yu,x_k,2}^i+{\cal{V}}_{yu,x_k,4}^i$, such that for any $y''\in{\cal{V}}_{y24,x_k}$, $T({\cal{V}}_{xu,x_k,1}^i+{\cal{V}}_{xu,x_k,2}^i,y'')=\sum_{y_i\in{\cal{V}}_{y}^{''}}T({\cal{V}}_{xu,x_k,1}^i+{\cal{V}}_{xu,x_k,2}^i,y_i)$ for some ${\cal{V}}_{y}^{''}={\cal{V}}_{y24,x_k}-y''+{\cal{V}}_{y1}^{''}$ with ${\cal{V}}_{y1}^{''}\subseteq{\cal{V}}_{yu,x_k,1}^i+{\cal{V}}_{yu,x_k,3}^i$. Let ${\cal{E}}_{u24,x_k}\subseteq{\cal{E}}_{u,x_k,2}^i+{\cal{E}}_{u,x_k,4}^i$ be the edge set with ${\cal{V}}_{y24,x_k}$ being their receiving nodes.
\end{lemma}
\begin{proof}
We know that rank$(T({\cal{V}}_{xu,x_k,1}^i+{\cal{V}}_{xu,x_k,2}^i,{\cal{V}}_{yu,x_k,1}^i+{\cal{V}}_{yu,x_k,2}^i+{\cal{V}}_{yu,x_k,3}^i+{\cal{V}}_{yu,x_k,4}^i))=|{\cal{V}}_{xu,x_k,1}^i|+|{\cal{V}}_{xu,x_k,2}^i|$. If rank$(T({\cal{V}}_{xu,x_k,1}^i+{\cal{V}}_{xu,x_k,2}^i,{\cal{V}}_{yu,x_k,1}^i+{\cal{V}}_{yu,x_k,3}^i))>|{\cal{V}}_{xu,x_k,1}^i|-|{\cal{V}}_{xu,x_k,4}^i|$, there must exist some $y'\in{\cal{V}}_{yu,x_k,2}^i+{\cal{V}}_{yu,x_k,4}^i$, such that $T({\cal{V}}_{xu,x_k,1}^i+{\cal{V}}_{xu,x_k,2}^i,y')=\sum_{y_i\in{\cal{V}}_{y}^{'}}T({\cal{V}}_{xu,x_k,1}^i+{\cal{V}}_{xu,x_k,2}^i,y_i)$ for some ${\cal{V}}_{y}^{'}\subseteq{\cal{V}}_{yu,x_k,1}^i+{\cal{V}}_{yu,x_k,2}^i+{\cal{V}}_{yu,x_k,3}^i+{\cal{V}}_{yu,x_k,4}^i-y'$. Let ${\cal{V}}_{y24,x_k}=\{y'+{\cal{V}}_{y}^{'}\}\cap\{{\cal{V}}_{yu,x_k,2}^i+{\cal{V}}_{yu,x_k,4}^i\}$. Obviously ${\cal{V}}_{y24,x_k}\neq\emptyset$ and $\sum_{y_i\in{\cal{V}}_{y24,x_k}}T({\cal{V}}_{xu,x_k,1}^i+{\cal{V}}_{xu,x_k,2}^i,y_i)=\sum_{y_i\in{\cal{V}}_{y}^{'}-{\cal{V}}_{y24,x_k}}T({\cal{V}}_{xu,x_k,1}^i+{\cal{V}}_{xu,x_k,2}^i,y_i)$. For any $y''\in{\cal{V}}_{y24,x_k}$, we have $T({\cal{V}}_{xu,x_k,1}^i+{\cal{V}}_{xu,x_k,2}^i,y'')=\sum_{y_i\in{\cal{V}}_{y24,x_k}-y''}T({\cal{V}}_{xu,x_k,1}^i+{\cal{V}}_{xu,x_k,2}^i,y_i)+\sum_{y_i\in{\cal{V}}_{y1}^{''}}T({\cal{V}}_{xu,x_k,1}^i+{\cal{V}}_{xu,x_k,2}^i,y_i)=\sum_{y_i\in{\cal{V}}_{y}^{''}}T({\cal{V}}_{xu,x_k,1}^i+{\cal{V}}_{xu,x_k,2}^i,y_i)$ with ${\cal{V}}_{y1}^{''}={\cal{V}}_{y}^{'}-{\cal{V}}_{y24,x_k}\subseteq{\cal{V}}_{yu,x_k,1}^i+{\cal{V}}_{yu,x_k,3}^i$ and ${\cal{V}}_{y}^{''}={\cal{V}}_{y24,x_k}-y''+{\cal{V}}_{y1}^{''}$.
\end{proof}

Assume rank$(T({\cal{V}}_{xu,x_k,1}^i+{\cal{V}}_{xu,x_k,2}^i,{\cal{V}}_{yu,x_k,1}^i+{\cal{V}}_{yu,x_k,3}^i))>|{\cal{V}}_{xu,x_k,1}^i|-|{\cal{V}}_{xu,x_k,4}^i|=K_i-K_{i4}$. Let $y''\in{\cal{V}}_{y24,x_k}$ be the last one in ${\cal{V}}_{y24,x_k}$ being added to the set ${\cal{V}}_{yu,x_k}^i$ before the current exploration of $x_k$ or the first one in ${\cal{V}}_{y24,x_k}$ being deleted from the set ${\cal{V}}_{yu,x_k}^i$ after the current exploration of $x_k$. By definition,
\begin{eqnarray}\label{eqn31}
T({\cal{V}}_{xu,x_k,1}^i+{\cal{V}}_{xu,x_k,2}^i,y'')=\sum_{y_i\in{\cal{V}}_{y}^{''}}T({\cal{V}}_{xu,x_k,1}^i+{\cal{V}}_{xu,x_k,2}^i,y_i)
\end{eqnarray}
for some ${\cal{V}}_{y}^{''}={\cal{V}}_{y24,x_k}-y''+{\cal{V}}_{y1}^{''}\subseteq{\cal{V}}_{yu,x_k}^i-y''$ with ${\cal{V}}_{y1}^{''}\subseteq{\cal{V}}_{yu,x_k,1}^i+{\cal{V}}_{yu,x_k,3}^i$. Let $x''$ be the corresponding transmitting node such that $(x'',y'')\in{\cal{E}}_{u24,x_k}$.

\begin{lemma}\label{lemm4}
Assume rank$(T({\cal{V}}_{xu,x_k,1}^i+{\cal{V}}_{xu,x_k,2}^i,{\cal{V}}_{yu,x_k,1}^i+{\cal{V}}_{yu,x_k,3}^i))>|{\cal{V}}_{xu,x_k,1}^i|-|{\cal{V}}_{xu,x_k,4}^i|=K_i-K_{i4}$ and let $(x'',y'')\in{\cal{E}}_{u24,x_k}$ be defined above. Then just after adding $y''$ or just before deleting $y''$, we have
\begin{eqnarray}\label{eqn32}
T({\cal{V}}_{xu,x'',1}^i+{\cal{V}}_{xu,x'',2}^i+x'',y'')=\sum_{y_i\in{\cal{V}}_{y}^{''}}T({\cal{V}}_{xu,x'',1}^i+{\cal{V}}_{xu,x'',2}^i+x'',y_i)
\end{eqnarray}
for the same ${\cal{V}}_{y}^{''}$ as in Equation (\ref{eqn31}). And when $x''$ is explored just before adding $y''$ along edge $(x'',y'')$ or just after deleting $y''$ along edge $(x'',y'')$, we have
rank$(T({\cal{V}}_{xu,x'',1}^i+{\cal{V}}_{xu,x'',2}^i+x'',{\cal{V}}_{yu,x''}^i))=|{\cal{V}}_{xu,x'',1}^i|+|{\cal{V}}_{xu,x'',2}^i|+1$.
\end{lemma}
\begin{proof}
Note that in equation (\ref{eqn32}), since $y''\in{\cal{V}}_{y24,x_k}$ with $(x'',y'')\in{\cal{E}}_{u,x_k}^i$ is the last one in the set ${\cal{V}}_{y24,x_k}$ being added to the set ${\cal{V}}_{yu,x_k}^i$ before the current exploration of $x_k$ or the first one in the set ${\cal{V}}_{y24,x_k}$ being deleted from the set ${\cal{V}}_{yu,x_k}^i$ after the current exploration of $x_k$, ${\cal{V}}_{y24,x_k}-y''$ is not changed and is also subset of ${\cal{V}}_{yu,x''}^i$. So in equation (\ref{eqn32}), ${\cal{V}}_y^{''}$ is the same as in (\ref{eqn31}) given that ${\cal{V}}_{y24,x_k}-y''$ is not changed and ${\cal{V}}_{y1}^{''}\subseteq{\cal{V}}_{yu,x_i,1}^i+{\cal{V}}_{yu,x_i,3}^i={\cal{V}}_{y{\cal{P}}1}^i+{\cal{V}}_{y{\cal{P}}3}^i$ is not changed, but the set ${\cal{V}}_{xu,x_k,1}^i+{\cal{V}}_{xu,x_k,2}^i$ are changing to ${\cal{V}}_{xu,x'',1}^i+{\cal{V}}_{xu,x'',2}^i$. When MDFS proceeds from the point of just after adding $(x'',y'')$ to ${\cal{E}}_u$ to the point of exploring $x_k$ or from the point of exploring $x_k$ to the point of just before deleting $(x'',y'')$ from ${\cal{E}}_u$, only those three forwarding moves in definition of MDFS are allowed. It is sufficient for us to show that any forwarding moves of MDFS or backtracking of these moves doesn't change the relationship in equation (\ref{eqn31}) so that equation (\ref{eqn32}) holds.

Let's first consider forwarding moves of MDFS. A type $1$ forwarding move along edge $(x,y)$ would change ${\cal{V}}_{xu,x}^i$ to ${\cal{V}}_{xu,x}^i+x$. Since ${\cal{V}}_{xspan}^{x}\subseteq{\cal{V}}_{xu,x,1}^i+{\cal{V}}_{xu,x,2}^i$ (based on Lemma \ref{lemm2}), the vector $T(x,{\cal{V}}_{yu,x}^i)$ is a linear combination of row vectors in $T({\cal{V}}_{xu,x,1}^i+{\cal{V}}_{xu,x,2}^i,{\cal{V}}_{yu,x}^i)$, so the relationship in (\ref{eqn31}) still holds while ${\cal{V}}_{xu,x,1}^i+{\cal{V}}_{xu,x,2}^i$ changes to ${\cal{V}}_{xu,x,1}^i+{\cal{V}}_{xu,x,2}^i+x$ in a type $1$ forwarding move. In a type $2$ forwarding move, ${\cal{V}}_{xu,x,1}^i+{\cal{V}}_{xu,x,2}^i$ changes to ${\cal{V}}_{xu,x,1}^i+{\cal{V}}_{xu,x,2}^i+x-x_i$ for some $x_i\in{\cal{V}}_{xspan}^{x}$. Again we have ${\cal{V}}_{xspan}^{x}\subseteq{\cal{V}}_{xu,x,1}^i+{\cal{V}}_{xu,x,2}^i$ (based on Lemma \ref{lemm2}). It is easy to see that the relationship in (\ref{eqn31}) still holds when ${\cal{V}}_{xu,x,1}^i+{\cal{V}}_{xu,x,2}^i$ changes to ${\cal{V}}_{xu,x,1}^i+{\cal{V}}_{xu,x,2}^i+x-x_i$. In a type $3$ forwarding move, some $y\in{\cal{V}}_{yu,x}^i\char92 {\cal{V}}_{y24,x_k}$ is deleted from ${\cal{V}}_{yu,x}^i$ which obviously doesn't affect the relationship in (\ref{eqn31}). Now let's consider backtracking moves of MDFS. Let equation (\ref{eqn31}) hold after a type $3$ forwarding move of MDFS along an edge $(x,y)$ for some $y\not\in{\cal{V}}_{y24,x_k}$. After the type $3$ forwarding move along $(x,y)$, MDFS will explore $x$ when it explores ${\cal{N}}(x)$. Again we have ${\cal{V}}_{xspan}^{x}\subseteq{\cal{V}}_{xu,x,1}^i+{\cal{V}}_{xu,x,2}^i$, so equation (\ref{eqn31}) should hold before the type $3$ forwarding move when ${\cal{V}}_{xu,x,1}^i+{\cal{V}}_{xu,x,2}^i$ was ${\cal{V}}_{xu,x,1}^i+{\cal{V}}_{xu,x,2}^i+x$. A proceeding type $2$ forwarding move before the current exploration of $x$ means ${\cal{V}}_{xu,x,1}^i+{\cal{V}}_{xu,x,2}^i$ was ${\cal{V}}_{xu,x,1}^i+{\cal{V}}_{xu,x,2}^i+x-x'$ with $x\in{\cal{V}}_{xspan}^{x'}$ before the move. Let equation (\ref{eqn31}) hold after a type $2$ forwarding move of MDFS along a path $P_{x'\rightarrow x}$ with ${\cal{V}}_{xu,x,1}^i+{\cal{V}}_{xu,x,2}^i$. Again we have ${\cal{V}}_{xspan}^{x}\subseteq{\cal{V}}_{xu,x,1}^i+{\cal{V}}_{xu,x,2}^i$ (based on Lemma \ref{lemm2}), so equation (\ref{eqn31}) holds with addition of row for $x$, which means that equation (\ref{eqn31}) holds before the type $2$ forwarding move when ${\cal{V}}_{xu,x,1}^i+{\cal{V}}_{xu,x,2}^i$ was ${\cal{V}}_{xu,x,1}^i+{\cal{V}}_{xu,x,2}^i+x-x'$. If the proceeding move is a type $1$ forwarding move, it means that ${\cal{V}}_{xu,x,1}^i+{\cal{V}}_{xu,x,2}^i$ was ${\cal{V}}_{xu,x,1}^i+{\cal{V}}_{xu,x,2}^i-x$, it is obvious that equation (\ref{eqn31}) holds before the type $1$ forwarding move with ${\cal{V}}_{xu,x,1}^i+{\cal{V}}_{xu,x,2}^i-x$ if it holds after the move with ${\cal{V}}_{xu,x,1}^i+{\cal{V}}_{xu,x,2}^i$. From above discussion, we conclude that equation (\ref{eqn32}) holds.

We know that after adding $(x'',y'')$ or before deleting $(x'',y'')$, rank$(T({\cal{V}}_{xu,x'',1}^i+{\cal{V}}_{xu,x'',2}^i+x'',{\cal{V}}_{yu,x''}^i+y''))=|{\cal{V}}_{xu,x'',1}^i|+|{\cal{V}}_{xu,x'',2}^i|+1$. Given that equation (\ref{eqn32}) holds, we have rank$(T({\cal{V}}_{xu,x'',1}^i+{\cal{V}}_{xu,x'',2}^i+x'',{\cal{V}}_{yu,x''}^i))=|{\cal{V}}_{xu,x'',1}^i|+|{\cal{V}}_{xu,x'',2}^i|+1$.
\end{proof}

\begin{lemma}\label{lem20}
For any $x_k$ that has been explored by MDFS in iteration $K+1$ of our algorithm, rank$(T({\cal{V}}_{xu,x_k,1}^i+{\cal{V}}_{xu,x_k,2}^i,{\cal{V}}_{yu,x_k,1}^i+{\cal{V}}_{yu,x_k,3}^i))=|{\cal{V}}_{xu,x_k,1}^i|-|{\cal{V}}_{xu,x_k,4}^i|=K_i-K_{i4}\geq0$.
\end{lemma}
\begin{proof}
We first prove $K_i-K_{i4}\geq0$. Assume $K_i-K_{i4}<0$. Then we must have rank$(T({\cal{V}}_{xu,x_k,1}^i+{\cal{V}}_{xu,x_k,2}^i,{\cal{V}}_{yu,x_k,1}^i+{\cal{V}}_{yu,x_k,3}^i))>|{\cal{V}}_{xu,x_k,1}^i|-|{\cal{V}}_{xu,x_k,4}^i|=K_i-K_{i4}$. Based on Lemma \ref{lemm3} and \ref{lemm4}, if rank$(T({\cal{V}}_{xu,x_k,1}^i+{\cal{V}}_{xu,x_k,2}^i,{\cal{V}}_{yu,x_k,1}^i+{\cal{V}}_{yu,x_k,3}^i))>|{\cal{V}}_{xu,x_k,1}^i|-|{\cal{V}}_{xu,x_k,4}^i|=K_i-K_{i4}$, then we will have rank$(T({\cal{V}}_{xu,x'',1}^i+{\cal{V}}_{xu,x'',2}^i+x'',{\cal{V}}_{yu,x''}^i))=|{\cal{V}}_{xu,x'',1}^i|+|{\cal{V}}_{xu,x'',2}^i|+1$ when $x''$ is explored just before adding $y''$ along edge $(x'',y'')$ or just after deleting $y''$ along edge $(x'',y'')$ for $x''$, $y''$ defined as in Lemma \ref{lemm4}. But based on Lemma \ref{lemm2}, we should have rank$(T({\cal{V}}_{xu,x'',1}^i+{\cal{V}}_{xu,x'',2}^i+x'',{\cal{V}}_{yu,x''}^i))=|{\cal{V}}_{xu,x'',1}^i|+|{\cal{V}}_{xu,x'',2}^i|$, which constitutes a contradiction. So we must have $K_i-K_{i4}\geq0$.

Assume rank$(T({\cal{V}}_{xu,x_k,1}^i+{\cal{V}}_{xu,x_k,2}^i,{\cal{V}}_{yu,x_k,1}^i+{\cal{V}}_{yu,x_k,3}^i))>|{\cal{V}}_{xu,x_k,1}^i|-|{\cal{V}}_{xu,x_k,4}^i|=K_i-K_{i4}\geq0$. Using a similar argument as above, we would arrive at a contradiction which means the assumption doesn't hold. So we must have rank$(T({\cal{V}}_{xu,x_k,1}^i+{\cal{V}}_{xu,x_k,2}^i,{\cal{V}}_{yu,x_k,1}^i+{\cal{V}}_{yu,x_k,3}^i))\leq|{\cal{V}}_{xu,x_k,1}^i|-|{\cal{V}}_{xu,x_k,4}^i|=K_i-K_{i4}$. Now together with Lemma \ref{lemm1}, we conclude that rank$(T({\cal{V}}_{xu,x_k,1}^i+{\cal{V}}_{xu,x_k,2}^i,{\cal{V}}_{yu,x_k,1}^i+{\cal{V}}_{yu,x_k,3}^i))=|{\cal{V}}_{xu,x_k,1}^i|-|{\cal{V}}_{xu,x_k,4}^i|=K_i-K_{i4}$.
\end{proof}

\begin{lemma}\label{lemm5}
For any $x_k$ with ${\cal{L}}(x_k)=i$ that has been explored by MDFS in iteration $K+1$ of our algorithm and any $y_j\in{\cal{V}}_{y\Omega_K3}^i$, rank$(T({\cal{V}}_{xu,x_k,1}^i+{\cal{V}}_{xu,x_k,2}^i+x_k,{\cal{V}}_{yu,x_k,1}^i+{\cal{V}}_{yu,x_k,3}^i+y_j))=$rank$(T({\cal{V}}_{xu,x_k,1}^i+{\cal{V}}_{xu,x_k,2}^i+x_k,{\cal{V}}_{yu,x_k,1}^i+{\cal{V}}_{yu,x_k,3}^i))=|{\cal{V}}_{xu,x_k,1}^i|-|{\cal{V}}_{xu,x_k,4}^i|=K_i-K_{i4}$.
\end{lemma}
\begin{proof}
First we prove rank$(T({\cal{V}}_{xu,x_k,1}^i+{\cal{V}}_{xu,x_k,2}^i+x_k,{\cal{V}}_{yu,x_k,1}^i+{\cal{V}}_{yu,x_k,3}^i))=K_i-K_{i4}$. Based on Lemma \ref{lemm2}, ${\cal{V}}_{xspan}^{x_k}\subseteq{\cal{V}}_{xu,x_k,1}^i+{\cal{V}}_{xu,x_k,2}^i$, so rank$(T({\cal{V}}_{xu,x_k,1}^i+{\cal{V}}_{xu,x_k,2}^i+x_k,{\cal{V}}_{yu,x_k,1}^i+{\cal{V}}_{yu,x_k,3}^i))=$rank$(T({\cal{V}}_{xu,x_k,1}^i+{\cal{V}}_{xu,x_k,2}^i,{\cal{V}}_{yu,x_k,1}^i+{\cal{V}}_{yu,x_k,3}^i))=|{\cal{V}}_{xu,x_k,1}^i|-|{\cal{V}}_{xu,x_k,4}^i|=K_i-K_{i4}$ based on Lemma \ref{lem20}.

Second we prove rank$(T({\cal{V}}_{xu,x_k,1}^i+{\cal{V}}_{xu,x_k,2}^i+x_k,{\cal{V}}_{yu,x_k,1}^i+{\cal{V}}_{yu,x_k,3}^i+y_j))=|{\cal{V}}_{xu,x_k,1}^i|-|{\cal{V}}_{xu,x_k,4}^i|=K_i-K_{i4}$. Given rank$(T({\cal{V}}_{xu,x_k,1}^i+{\cal{V}}_{xu,x_k,2}^i+x_k,{\cal{V}}_{yu,x_k,1}^i+{\cal{V}}_{yu,x_k,3}^i))=|{\cal{V}}_{xu,x_k,1}^i|-|{\cal{V}}_{xu,x_k,4}^i|=K_i-K_{i4}$, rank$(T({\cal{V}}_{xu,x_k,1}^i+{\cal{V}}_{xu,x_k,2}^i+x_k,{\cal{V}}_{yu,x_k,1}^i+{\cal{V}}_{yu,x_k,3}^i+y_j))$ equals either $|{\cal{V}}_{xu,x_k,1}^i|-|{\cal{V}}_{xu,x_k,4}^i|$ or $|{\cal{V}}_{xu,x_k,1}^i|-|{\cal{V}}_{xu,x_k,4}^i|+1$. Assume
\begin{eqnarray}
{\mbox{rank}}(T({\cal{V}}_{xu,x_k,1}^i+{\cal{V}}_{xu,x_k,2}^i+x_k,{\cal{V}}_{yu,x_k,1}^i+{\cal{V}}_{yu,x_k,3}^i+y_j))=|{\cal{V}}_{xu,x_k,1}^i|-|{\cal{V}}_{xu,x_k,4}^i|+1\label{eqn33}
\end{eqnarray}
From Lemma \ref{lemm2}, we have
\begin{eqnarray}
&&{\mbox{rank}}(T({\cal{V}}_{xu,x_k,1}^i+{\cal{V}}_{xu,x_k,2}^i+x_k,{\cal{V}}_{yu,x_k}^i+y_j))\\
&=&{\mbox{rank}}(T({\cal{V}}_{xu,x_k,1}^i+{\cal{V}}_{xu,x_k,2}^i,{\cal{V}}_{yu,x_k}^i+y_j))=|{\cal{V}}_{xu,x_k,1}^i|+|{\cal{V}}_{xu,x_k,2}^i|\label{eqn37}
\end{eqnarray}
and
\begin{eqnarray}
&&T(x_k,{\cal{V}}_{yu,x_k}^i+y_j)\\
&=&\sum_{x'\in{\cal{V}}_{xspan}^{x_k}\subseteq{\cal{V}}_{xu,x_k,1}^i+{\cal{V}}_{xu,x_k,2}^i}T(x',{\cal{V}}_{yu,x_k}^i+y_j)\label{eqn35}
\end{eqnarray}
From (\ref{eqn33}) and (\ref{eqn35}), we have
\begin{eqnarray}
{\mbox{rank}}(T({\cal{V}}_{xu,x_k,1}^i+{\cal{V}}_{xu,x_k,2}^i,{\cal{V}}_{yu,x_k,1}^i+{\cal{V}}_{yu,x_k,3}^i+y_j))=|{\cal{V}}_{xu,x_k,1}^i|-|{\cal{V}}_{xu,x_k,4}^i|+1\label{eqn36}
\end{eqnarray}
From (\ref{eqn37}) and (\ref{eqn36}), we have
\begin{eqnarray}
&&T({\cal{V}}_{xu,x_k,1}^i+{\cal{V}}_{xu,x_k,2}^i,y_j)\\
&=&\sum_{y_{13}\in{\cal{V}}_{y13,x_k}}T({\cal{V}}_{xu,x_k,1}^i+{\cal{V}}_{xu,x_k,2}^i,y_{13})+\sum_{y_{24}\in{\cal{V}}_{y24,x_k}}T({\cal{V}}_{xu,x_k,1}^i+{\cal{V}}_{xu,x_k,2}^i,y_{24})\label{eqn38}
\end{eqnarray}
for some ${\cal{V}}_{y13,x_k}\subseteq{\cal{V}}_{yu,x_k,1}^i+{\cal{V}}_{yu,x_k,3}^i$ and some nonempty ${\cal{V}}_{y24,x_k}\subseteq{\cal{V}}_{yu,x_k,2}^i+{\cal{V}}_{yu,x_k,4}^i$.
From (\ref{eqn35}) and (\ref{eqn38}), we have
\begin{eqnarray}
&&T({\cal{V}}_{xu,x_k,1}^i+{\cal{V}}_{xu,x_k,2}^i+x_k,y_j)\\
&=&\sum_{y_{13}\in{\cal{V}}_{y13,x_k}}T({\cal{V}}_{xu,x_k,1}^i+{\cal{V}}_{xu,x_k,2}^i+x_k,y_{13})+\sum_{y_{24}\in{\cal{V}}_{y24,x_k}}T({\cal{V}}_{xu,x_k,1}^i+{\cal{V}}_{xu,x_k,2}^i+x_k,y_{24})
\end{eqnarray}
for the same ${\cal{V}}_{y13,x_k}$ and ${\cal{V}}_{y24,x_k}$ in (\ref{eqn38}).

Let $y''\in{\cal{V}}_{y24,x_k}$ be the last one in ${\cal{V}}_{y24,x_k}$ being added to the set ${\cal{V}}_{yu,x_k}^i$ before the current exploration of $x_k$ or the first one in ${\cal{V}}_{y24,x_k}$ being deleted from the set ${\cal{V}}_{yu,x_k}^i$ after the current exploration of $x_k$ and $(x'',y'')\in{\cal{E}}_{u,x_k}^i$. Then using a similar argument as that in Lemma \ref{lemm4}, we have
rank$(T({\cal{V}}_{xu,x'',1}^i+{\cal{V}}_{xu,x'',2}^i+x'',{\cal{V}}_{yu,x''}^i+y_j))=|{\cal{V}}_{xu,x'',1}^i|+|{\cal{V}}_{xu,x'',2}^i|+1$ when $x''$ is explored just before adding $y''$ along edge $(x'',y'')$ or just after deleting $y''$ along edge $(x'',y'')$, but it's a contradiction with Lemma \ref{lemm2}. So it must be rank$(T({\cal{V}}_{xu,x_k,1}^i+{\cal{V}}_{xu,x_k,2}^i+x_k,{\cal{V}}_{yu,x_k,1}^i+{\cal{V}}_{yu,x_k,3}^i+y_j))=|{\cal{V}}_{xu,x_k,1}^i|-|{\cal{V}}_{xu,x_k,4}^i|=K_i-K_{i4}$. \end{proof}

\begin{lemma}\label{lemm6}
rank$(T_{\Omega_K}^i)=$rank$(T({\cal{V}}_{x\Omega_K}^i,{\cal{V}}_{y\Omega_K}^i))=K_i-K_{i4},1\leq i\leq L-1$.
\end{lemma}
\begin{proof}
If ${\cal{E}}_{\Omega_K}^i=\emptyset$, i.e., the cut $\Omega_K$ and layer cut $i$ has no intersection, then rank$(T_{\Omega_K}^i)=$ rank$(T({\cal{V}}_{x\Omega_K}^i,{\cal{V}}_{y\Omega_K}^i))=K_i-K_{i4}=0$ holds. Next assume that ${\cal{E}}_{\Omega_K}^i\neq\emptyset$. Lemma \ref{lemm5} says for any $x_k$ with ${\cal{L}}(x_k)=i$ that has been explored by MDFS in iteration $K+1$ of our algorithm and any $y_j\in{\cal{V}}_{y\Omega_K3}^i$, rank$(T({\cal{V}}_{xu,x_k,1}^i+{\cal{V}}_{xu,x_k,2}^i+x_k,{\cal{V}}_{yu,x_k,1}^i+{\cal{V}}_{yu,x_k,3}^i+y_j))=K_i-K_{i4}$ and rank$(T({\cal{V}}_{xu,x_k,1}^i+{\cal{V}}_{xu,x_k,2}^i+x_k,{\cal{V}}_{yu,x_k,1}^i+{\cal{V}}_{yu,x_k,3}^i))=K_i-K_{i4}$. Lemma \ref{lemm2} says for any $y_j\in{\cal{V}}_{y\Omega_K3}^i$, $T(x_k,{\cal{V}}_{yu,x_k}^i+y_j)=\sum_{x'\in{\cal{V}}_{xspan}^{x_k}\subseteq{\cal{V}}_{xu,x_k,1}^i+{\cal{V}}_{xu,x_k,2}^i}T(x',{\cal{V}}_{yu,x_k}^i+y_j)$. Based on these two Lemmas, it's easy to conclude that
\begin{eqnarray}
&&{\mbox{rank}}(T({\cal{V}}_{xu,x_k,1}^i+{\cal{V}}_{xu,x_k,2}^i+x_k,{\cal{V}}_{yu,x_k,1}^i+{\cal{V}}_{yu,x_k,3}^i+{\cal{V}}_{y\Omega_K3}^i))\\ &=&{\mbox{rank}}(T({\cal{V}}_{xu,x_k,1}^i+{\cal{V}}_{xu,x_k,2}^i+x_k,{\cal{V}}_{y\Omega_K}^i))=K_i-K_{i4}\label{eqn39}\\
&=&{\mbox{rank}}(T({\cal{V}}_{xu,x_k,1}^i+{\cal{V}}_{xu,x_k,2}^i,{\cal{V}}_{y\Omega_K}^i))=K_i-K_{i4}\label{eqn41}
\end{eqnarray}
and
\begin{eqnarray}
T(x_k,{\cal{V}}_{y\Omega_K}^i)=\sum_{x'\in{\cal{V}}_{xspan}^{x_k}\subseteq{\cal{V}}_{xu,x_k,1}^i+{\cal{V}}_{xu,x_k,2}^i}T(x',{\cal{V}}_{y\Omega_K}^i)\label{eqn40} \end{eqnarray}
Equations (\ref{eqn39}) (\ref{eqn41}) and (\ref{eqn40}) hold for any $x_k$ that has been explored by MDFS in iteration $K+1$ of our algorithm.

Let $x^q,1\leq q\leq Q$ be the $q$th transmitting nodes in layer $i$ that has been explored by MDFS in iteration $K+1$ in our algorithm. Note that since some transmitting nodes may be explored more than once, $x^q$ may not be distinct but $Q$ is finite.
We claim that
\begin{eqnarray}
{\mbox{rank}}(T({\cal{V}}_{x{\cal{P}}1}^i+{\cal{V}}_{x{\cal{P}}2}^i+\sum_{k=1}^{q}x^k,{\cal{V}}_{y\Omega_K}^i))={\mbox{rank}}(T({\cal{V}}_{x{\cal{P}}1}^i+{\cal{V}}_{x{\cal{P}}2}^i,{\cal{V}}_{y\Omega_K}^i))=K_i-K_{i4}\label{eqn46}
\end{eqnarray}
for $1\leq q\leq Q$. When $x^1$ is explored, ${\cal{V}}_{xu,x^1,1}^i={\cal{V}}_{x{\cal{P}}1}^i$ and ${\cal{V}}_{xu,x^1,2}^i={\cal{V}}_{x{\cal{P}}2}^i$. From (\ref{eqn39}) (\ref{eqn41}), we have
\begin{eqnarray}
{\mbox{rank}}(T({\cal{V}}_{x{\cal{P}}1}^i+{\cal{V}}_{x{\cal{P}}2}^i+x^1,{\cal{V}}_{y\Omega_K}^i))={\mbox{rank}}(T({\cal{V}}_{x{\cal{P}}1}^i+{\cal{V}}_{x{\cal{P}}2}^i,{\cal{V}}_{y\Omega_K}^i))=K_i-K_{i4}\label{eqn42}
\end{eqnarray}
When $x^2$ is explored,
\begin{eqnarray}
{\cal{V}}_{xu,x^2,1}^i+{\cal{V}}_{xu,x^2,2}^i\subseteq{\cal{V}}_{xu,x^1,1}^i+{\cal{V}}_{xu,x^1,2}^i+x^1\label{eqn45}
\end{eqnarray}
and from (\ref{eqn39}) (\ref{eqn41})
\begin{eqnarray}
{\mbox{rank}}(T({\cal{V}}_{xu,x^2,1}^i+{\cal{V}}_{xu,x^2,2}^i+x^2,{\cal{V}}_{y\Omega_K}^i))={\mbox{rank}}(T({\cal{V}}_{xu,x^2,1}^i+{\cal{V}}_{xu,x^2,2}^i,{\cal{V}}_{y\Omega_K}^i))=K_i-K_{i4}\label{eqn43}
\end{eqnarray}
From (\ref{eqn42}) to (\ref{eqn43}), we conclude that rank$(T({\cal{V}}_{x{\cal{P}}1}^i+{\cal{V}}_{x{\cal{P}}2}^i+x^1+x^2,{\cal{V}}_{y\Omega_K}^i))=K_i-K_{i4}$. Use induction on $x^q,1\leq q\leq Q$, just as we did for $x^2$, we conclude that (\ref{eqn46}) holds for any $q$,$1\leq q\leq Q$. We know that ${\cal{V}}_{x{\cal{P}}1}^i+{\cal{V}}_{x{\cal{P}}2}^i+\sum_{k=1}^{Q}x^k={\cal{V}}_{x\Omega_K}^i$, so when $q=Q$ equation (\ref{eqn46}) means that rank$(T_{\Omega_K}^i)=$rank$(T({\cal{V}}_{x\Omega_K}^i,{\cal{V}}_{y\Omega_K}^i))=K_i-K_{i4},1\leq i\leq L-1$.
\end{proof}

\begin{lemma}\label{lemm7}
\begin{eqnarray}
\sum_{i=1}^{L-1}K_i-\sum_{i=1}^{L-1}K_{i4}=K
\end{eqnarray}
\end{lemma}
\begin{proof}
By definition, $\sum_{i=1}^{L-1}K_i$ is the total times the paths in ${\cal{P}}$ cross the cut $\Omega_K$. Since each path in ${\cal{P}}$ crosses $\Omega_K$ at least once and it's possible that some of these paths may cross $\Omega_K$ more than once, we have $\sum_{i=1}^{L-1}K_i\geq K$. For $\forall P\in{\cal{P}}$, let $k_{p}$ be the times $P$ goes from $\Omega_K$ to $\Omega_K^c$ and $k_{P}^{'}$ be the times $P$ goes from $\Omega_K^c$ to $\Omega_K$. Given $S\in\Omega_K$ and $D\in\Omega_K^c$, it must be true that $k_{P}-k_{P}^{'}=1$ and $\sum_{P\in{\cal{P}}}(k_{P}-k_{P}^{'})=|{\cal{P}}|=K$. By definition, the times that $\forall P\in{\cal{P}}$ goes from $\Omega_K$ to $\Omega_K^c$ is counted in $\sum_{i=1}^{L-1}K_i$ and the times that $\forall P\in{\cal{P}}$ goes from $\Omega_K^c$ to $\Omega_K$ is counted in $\sum_{i=1}^{L-1}K_{i4}$, i.e., $\sum_{P\in{\cal{P}}}(k_{P}-k_{P}^{'})=\sum_{i=1}^{L-1}K_i-\sum_{i=1}^{L-1}K_{i4}$. So we have $\sum_{i=1}^{L-1}K_i-\sum_{i=1}^{L-1}K_{i4}=K$.
\end{proof}

Based on Lemma \ref{lemm6}, we have
\begin{eqnarray}
{\mbox{rank}}(T({\cal{E}}_{{\Omega_K}}))=\sum_{i=1}^{L-1}{\mbox{rank}}(T_{\Omega_K}^i)=\sum_{i=1}^{L-1}(K_i-K_{i4})=\sum_{i=1}^{L-1}K_i-\sum_{i=1}^{L-1}K_{i4}
\end{eqnarray}
In Lemma \ref{lemm7}, we show that $\sum_{i=1}^{L-1}K_i-\sum_{i=1}^{L-1}K_{i4}=K$ which means rank$(T({\cal{E}}_{{\Omega_K}}))=K$ and this concludes our proof for Theorem \ref{theorem1}.

\end{proof}

Theorem \ref{theorem2} proves that our algorithm terminates in finite time. Theorem \ref{theorem1} proves that our algorithm finds $C$ linearly independent $S-D$ paths (${\cal{P}}$) where $C$ is the unicast capacity of the underlying deterministic relay network. Lemma \ref{lemm9} shows that these $C$ paths in ${\cal{P}}$ correspond to a capacity-achieving transmission scheme. They consist the complete proof of correctness for our algorithm for finding the unicast capacity of any linear layered deterministic wireless relay network.

An arbitrary deterministic relay network ${\cal{G}}$ can be unfolded over time to create a layered deterministic network ${\cal{G}}_L$ through time-expansion technique \cite{amir2007_deterministicmodel}\cite{amir2007_wirelessnetworkinfoflow}. The transmission scheme in ${\cal{G}}_L$ identified by our algorithm corresponds to some equivalent transmission scheme in ${\cal{G}}$ maybe time-variant. It means that our algorithm works for finding the unicast capacity of an arbitrary linear deterministic relay network.

\section{Conclusions\label{conclusion}}

The deterministic channel model for wireless relay networks has been a useful model for studying the capacity and capacity-approaching transmission schemes for underlying networks. In this paper, we proposed a fast algorithm for finding the unicast capacity of any given linear deterministic wireless relay network. Our algorithm finds the maximum number of linearly independent paths for the deterministic relay network and these paths correspond to a capacity-achieving transmission scheme. The essential component of our algorithm is a modified depth-first search algorithm developed for linear deterministic wireless relay networks. The proof of correctness for the algorithm is given which guarantees that our algorithm works in universal cases. Compared with previous results on solving the same problem, our algorithm prevails with significantly lower computational complexity. Moreover, the development of the modified depth-first search algorithm is based on a very intuitive idea, that is, to build up the path by adding edges while avoiding the linear dependency by using rank check at the same time.

\bibliographystyle{IEEEtran}
\bibliography{FastAlgorithmforFindingUnicastCapacityofLinearDeterministicWirelessRelayNetworks}

\end{document}